\documentclass[journal]{IEEEtran}

\usepackage{amsmath,amssymb,graphicx}
\usepackage{amsthm}
\usepackage{float}

\title{Compressed Beamforming in Ultrasound Imaging}
\author{\IEEEauthorblockN{Noam Wagner, Yonina C. Eldar and  Zvi Friedman}

\thanks{This work was supported in part by a Magneton grant from the Israel Ministry of Industry and Trade.}
\thanks{N. Wagner and Y. C. Eldar are both with the Department of Electrical Engineering, Technion-Israel Institute of Technology, Haifa 32000, Israel (e-mails: noamwa@techunix.technion.ac.il, yonina@ee.technion.ac.il).}
\thanks{Z. Friedman is with the Department of Biomedical Engineering, Technion-Israel Institute of Technology, Haifa 32000, Israel.  He is also with GE Health-care, Haifa, Israel (e-mail: zvi.friedman@med.ge.com).}}

\newtheorem{mydef}{Proposition}

\begin{document}
%
\maketitle
\begin{abstract}
Emerging sonography techniques often require increasing the number of transducer elements involved in the imaging process.   Consequently, larger amounts of data must be acquired and processed.  The significant growth in the amounts of data affects both machinery size and power consumption.  Within the classical sampling framework, state of the art systems reduce processing rates by exploiting the bandpass bandwidth of the detected signals.  It has been recently shown, that a much more significant sample-rate reduction may be obtained, by treating ultrasound signals within the Finite Rate of Innovation framework.  These ideas follow the spirit of Xampling, which combines classic methods from sampling theory with recent developments in Compressed Sensing.  Applying such low-rate sampling schemes to individual transducer elements, which detect energy reflected from biological tissues, is limited by the noisy nature of the signals.  This often results in erroneous parameter extraction, bringing forward the need to enhance the SNR of the low-rate samples.  In our work, we achieve SNR enhancement, by beamforming the sub-Nyquist samples obtained from multiple elements. We refer to this process as ``compressed beamforming".  Applying it to cardiac ultrasound data, we successfully image macroscopic perturbations, while achieving a nearly eight-fold reduction in sample-rate, compared to standard techniques.   

\end{abstract}
\section{Introduction}
\label{sec:01}  
Diagnostic sonography allows visualization of body tissues, by radiating them with acoustic energy pulses, which are transmitted from an array of transducer elements.  The image typically comprises multiple scanlines, each constructed by integrating data collected by the transducers, following the transmission of an energy pulse along a narrow beam.  As the pulse propagates, echoes are scattered by density and propagation-velocity perturbations in the tissue~\cite{Jensen01}, and detected by the transducer elements.  Averaging the detected signals, after their alignment with appropriate time-varying delays, allows localization of the scattering structures, while improving the Signal to Noise Ratio (SNR)~\cite{Szabo01}.   The latter process is referred to as beamforming.  Performed digitally, beamforming requires that the analog signals, detected by the transducers, first be sampled. Confined to classic Nyquist-Shannon sampling theorem~\cite{Shannon01}, the sampling rate must be at least twice the bandwidth, in order to avoid aliasing.  

As imaging techniques develop, the amount of elements involved in each imaging cycle typically increases.  Consequently, the rates of data which need to be transmitted from the system front-end, and then processed by the beamformer, grow significantly.  The growth in transmission and processing rates inevitably effects both machinery size and power consumption.   Consequently, in recent years there has been growing interest in reducing the amounts of data as close as possible to the system front-end.  In fact, such reduction is already possible within the classical sampling framework:  state of the art devices digitally downsample the data at the front-end, by exploiting the fact that the signal is modulated onto a carrier, so that the spectrum essentially occupies only a portion of its entire base-band bandwidth.  The preliminary sample rate remains unchanged, since the demodulation is performed in the digital domain.  Nevertheless, a key to significant data compression lies beyond the classical sampling framework.  

Indeed, the emerging Compressive Sensing (CS) framework~\cite{Candes01, Davenport01} states, that sparse signals may be accurately reconstructed from a surprisingly small amount of coefficients.  Complementary ideas rise from the Finite Rate of Innovation (FRI) framework~\cite{Vetterli01}, in which the signal is assumed to have a finite number of degrees of freedom per unit time.  Many classes of FRI signals can be recovered from samples taken at the rate of innovation~\cite{Michaeli01}.  For a detailed review of previously proposed FRI methods, the reader is referred to \cite{Uriguen01}.  Combining the latter notions with classical sampling methods, the developing Xampling framework~\cite{Mishali01,Mishali05,Mishali06} involves methods for fully capturing the information carried by an analog signal, by sampling it far below the Nyquist-rate.  

Following the spirit of Xampling, Tur et. al. proposed in~\cite{Tur01}, that ultrasound signals be described within the FRI framework.  Explicitly, they assume that these signals, formed by scattering of a transmitted pulse from multiple reflectors, may be modeled by a relatively small number of pulses, all replicas of some known pulse shape.  Denoting the number of reflected pulses by $L$, and the signal's finite temporal support by $\left[0,T\right)$, the detected signal is completely defined by $2L$ degrees of freedom, corresponding to the replicas' unknown time delays and amplitudes.  Based on~\cite{Vetterli01}, the authors formulate the relationship between the signal's Fourier series coefficients, calculated with respect to $\left[0,T\right)$, and its unknown parameters, in the form of a spectral analysis problem.  The latter may be solved using existing techniques, given a subset of Fourier series coefficients, with a minimal cardinality of $2L$.  The sampling scheme is thus reduced to the problem of extracting a small subset of the detected signal's frequency samples.  Two robust schemes are derived in \cite{Tur01,Gedalyahu01}, extracting such a set of coefficients from samples of the signal, taken at sub-Nyquist rates.  The system presented in \cite{Tur01} employs a single processing channel, in which the analog signal is filtered by an appropriate sampling kernel and then sampled with a standard low-rate analog to digital converter (ADC).  The method of~\cite{Gedalyahu01} employs multiple processing channels, each comprising a modulator and an integrator.  These approaches were shown to be more robust than previous FRI techniques and also allow for arbitrary pulse shapes.    

The initial motivation for our work stems from the need to translate the ultrasound Xampling scheme proposed in \cite{Tur01}, into one which achieves the final goal of reconstructing a two-dimensional ultrasound image, by integrating data sampled at multiple transducer elements.  In conventional ultrasound imaging, such integration is achieved by the beamforming process.  The question is how may we implement beamforming, using samples of the detected signals taken at sub-Nyquist rates.    

A straightforward approach is to replace the Nyquist-rate sampling mechanism, utilized in each receiver element, by an  FRI Xampling scheme.  Having estimated the parametric representation of the signal detected in each individual element, we could reconstruct it digitally.  The reconstructed signals can then be further processed via beamforming.  However, the nature of ultrasound signals reflected from real tissues, makes such an approach impractical.  This is mainly due to the detected signals' poor SNR,  which results in erroneous parameter extraction by the Xampling scheme, applied to each element independently.  

Our approach is to generalize the FRI Xampling scheme proposed in~\cite{Gedalyahu01}, such that it integrates beamforming into the low-rate sampling process.  The result is equivalent to that obtained by Xampling the beamformed signal, which exhibits significantly better SNR.  Furthermore,  beamforming practically implies that the array of receivers is dynamically focused along a single scanline.  Consequently, the resulting signal depicts reflections originating in the intersection of the radiated medium with a vary narrow beam.  Such a signal better suits the FRI model proposed in~\cite{Tur01}, which assumes the reflections to be caused by isolated, point-like scatterers.  We refer to our scheme by the term  compressed beamforming, as it transforms the beamforming operator into the compressed domain~\cite{Wagner02,Wagner03}.  Applied to real cardiac ultrasound data obtained from a GE breadboard ultrasonic scanner, our approach successfully images macroscopic perturbations in the tissue while achieving a nearly eight-fold reduction in sampling rate, compared to standard imaging techniques.
 
The paper is organized as follows: in Section \ref{sec:02}, we summarize the general principles of beamforming in ultrasound imaging. In Section \ref{sec:03} we outline the FRI model and its contribution to sample rate reduction in the ultrasound context.  We motivate compressed beamforming  in Section \ref{sec:04}, considering the nature of ultrasound signals reflected from biological tissues.  Beamforming and FRI Xampling are combined in Section \ref{sec:05}, where we propose that the signal obtained by beamforming may be treated within the FRI framework.  Following this observation, we derive our first compressed beamforming scheme, which operates on low-rate samples taken at the individual receivers.  This approach is then further simplified in Section \ref{sec:06}.  In Section \ref{sec:07} we focus on image reconstruction from the parametric representation obtained by either Xampling scheme.  In this context, we generalize the signal model proposed in \cite{Tur01}, allowing additional unknown phase shifts of the detected pulses. We then discuss an alternative recovery approach, based on CS.  Simulations comparing the performance of several recovery methods are provided in Section \ref{sec:08}.  Finally, experimental results obtained for cardiac ultrasound data are presented in Section \ref{sec:09}.           

\section{Beamforming in Ultrasound Imaging}
\label{sec:02}
\begin{figure}
\begin{minipage}[b]{1.00\linewidth}
 \centering
  \centerline{\includegraphics[width=5.5cm]{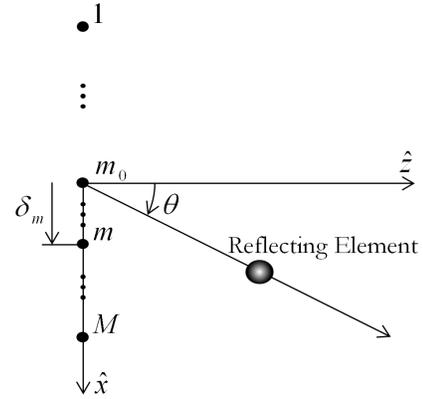}}
\end{minipage}
\caption{Imaging setup: $M$ receivers are aligned along the $\hat{x}$ axis.  The origin is set at the position of the reference receiver, denoted $m_0$.  $\delta_m$ denotes the distance measured from the reference receiver to the $m$th receiver.  The imaging cycle begins when an acoustic pulse is transmitted at direction $\theta$.  Echoes are then reflected from perturbations in the radiated medium.}
\label{Fig:01}
\end{figure}
In this section, we describe a typical B-mode imaging cycle, focusing on the beamforming process, carried out during the reception phase.  The latter constitutes a significant block in ultrasound imaging, and plays a major role in our proposed FRI Xampling scheme.  

Consider the array depicted in Fig. \ref{Fig:01}, comprising $M$ transducer elements, aligned along the $\hat{x}$ axis. Denote by $\delta_m$ the distance from the $m$th element to the reference receiver $m_0$, used as the  origin, namely $\delta_{m_0}=0$.  The imaging cycle begins when, at time $t=0$, the array transmits acoustic energy into the tissue.  Subsequently, the elements detect echoes, which originate in density and propagation-velocity perturbations, characterizing the radiated medium.   Denote by $\varphi_m\left(t\right)$ the signal detected by the $m$th receiver.   The acoustic reciprocity theorem~\cite{Kinsler01} suggests, that we may use  the signals detected by multiple transducer elements, in order to probe arbitrary coordinates for reflected energy.  Namely, by combining the detected signals with appropriate time delays, echoes scattered from a chosen coordinate will undergo constructive interference, whereas those originating off this coordinate will be attenuated, due to destructive interference. 

In practice, the array cannot effectively radiate the entire medium simultaneously.  Instead, a pulse of energy is conducted along a relatively narrow beam, whose central axis forms an angle $\theta$ with the $\hat{z}$ axis.  Focusing the energy pulse along such a beam is achieved by applying appropriate time delays to modulated acoustic pulses, transmitted from multiple array elements.  Rather than arbitrarily probing the radiated tissue, we are now forced to adjust the probed coordinate in time, in coordination with the propagation of the transmitted energy.  This practically implies that, combining the detected signals with appropriate time-varying delays, we may obtain a signal, which depicts the intensity of the energy reflected from each point along the central transmission axis.  Throughout the rest of this section, we derive an explicit expression for creating this beamformed signal.    

Assume that the energy pulse, transmitted at $t=0$, propagates at velocity $c$ in the direction $\theta$. At time $t\geq0$, the pulse crosses the coordinate  $\left(x,z\right)=\left(ct\sin\theta,ct\cos\theta\right)$.  Consider a potential reflection, originating in this coordinate, and arriving at the $m$th element.  The distance traveled by such a reflection is: 
\begin{equation}
\label{E:02}
d_m(t;\theta)=\sqrt{\left(ct\cos\theta\right)^2+\left(\delta_m-ct\sin\theta\right)^2}.
\end{equation}  
The time in which the reflection crosses this distance is $d_m\left(t;\theta\right)/c$, so that it reaches the receiver element at time
\begin{equation}
\label{E:03}
\begin{split}
{\hat{\tau}}_m(t;\theta)&=t+\frac{d_m\left(t;\theta\right)}{c}.
\end{split}
\end{equation}  
It is readily seen that ${\hat{\tau}}_{m_0}\left(t;\theta\right)=2t$.  Hence, in order to align the reflection detected in the $m$th receiver with the one detected in  the reference receiver, we need to apply a delay to $\varphi_m\left(t\right)$, such that the resulting signal, ${\hat{\varphi}}_m\left(t;\theta\right)$, satisfies ${\hat{\varphi}}_m\left(2t;\theta\right)= \varphi_m\left({\hat{\tau}}_m\left(t;\theta\right)\right)$.  Denoting $\tau_m\left(t;\theta\right)={\hat{\tau}}_m\left(t/2;\theta\right)$, and using \eqref{E:02}, we obtain the following distorted signal for $t\geq0$: 
\begin{equation}
\label{E:04}
\begin{split}
{\hat{\varphi}}_m\left(t;\theta\right)&=\varphi_m\left({{\tau}}_m\left(t;\theta\right)\right),\\
{\tau}_m\left(t;\theta\right)&=\frac{1}{2}\left(t+\sqrt{t^2-4\gamma_mt\sin\theta+4\gamma_m^2}\right),
\end{split}
\end{equation}
with $\gamma_m=\delta_m/c$.  The aligned signals may now be averaged, resulting in the beamformed signal
\begin{equation}
\label{E:05}
\Phi\left(t;\theta\right)=\frac{1}{M}\sum_{m=1}^M{{\hat{\varphi}}_m\left(t;\theta\right)},
\end{equation}     
which exhibits enhanced SNR compared to $\left\{{\hat{\varphi}}_m\left(t;\theta\right)\right\}_{m=1}^M$.  Furthermore, by its construction,  $\Phi\left(t;\theta\right)$ represents, for every $t\geq0$, the intensity which was measured when focusing the array to ${\bf{p}}\left(t\right)=\left(ct/2\sin\theta,ct/2\cos\theta\right)$.  Therefore, it may eventually be translated into an intensity pattern, plotted along the corresponding ray.   

Although defined over continuous time, ultrasound systems perform the process formulated in \eqref{E:04}-\eqref{E:05} in the digital domain, requiring that the analog signals $\varphi_m\left(t\right)$ first be sampled.  Confined to the classic Nyquist-Shannon sampling theorem, these systems sample the signals at twice their baseband bandwidth, in order to avoid aliasing.  The detected signals typically occupy only a portion of their baseband bandwidth.  Exploiting this fact, some state of the art systems manage to reduce the amount of samples transmitted from the front-end, by down-sampling the data, after demodulation and low-pass filtering.  However, since such operations are carried out digitally, the preliminary sampling-rate remains unchanged.  

To conclude this section, we evaluate the nominal number of samples needed to be taken from each active receiver element in order to obtain a single scanline using standard imaging techniques. Consider an ultrasound system which images to a nominal depth of $r=16 \mbox{cm}$.  The velocity at which the pulse propagates, $c$, varies between $1446\mbox{m/sec}$ (fat) to $1566\mbox{m/sec}$ (spleen)~\cite{Jensen02}. An average value of $1540\mbox{m/sec}$ is assumed by scanners for processing purposes, such that the duration of the detected signal is $T=2r/c\approx 210\mu\mbox{sec}$.  The signal's baseband bandwidth requires a nominal sampling rate of $f_s=16\mbox{Mhz}$, resulting in an overall number of $Tf_s=3360$ real-valued samples.  Assuming that the signal's passband bandwidth is only $4\mbox{MHz}$, the data sampled at Nyquist-rate may be finally down-sampled to 1680 real-valued samples.  These samples, taken from all active receivers, are now processed, according to \eqref{E:04}-\eqref{E:05}, in order to construct the beamformed signal.  Since standard imaging devices carry out beamforming by applying delay and sum operations to the sampled data, the amount of operations required for generating a single scanline is directly related to the sample rate.  

Regardless of our computational power, physical constraints imply that the time required for constructing a single scanline is at least $T$.  This takes into account the round-trip time required for the transmitted pulse to penetrate the entire imaging depth, and for the resulting echoes to cross a similar distance back to the array.  Nevertheless, sufficient computational power may allow construction of several scanlines, within that same time interval, increasing the overall imaging rate.  By using compressed beamforming, we aim at capturing significant information in the imaging plane, while reducing the sampling rate and consequently the processing rate.  This, in turn, may improve the existing trade-off between imaging rates and both machinery size and power consumption.   
  
\section{Sample Rate Reduction Using the FRI Model}
\label{sec:03}
In a pioneer attempt to implement Xampling methodology in the context of ultrasound imaging, \cite{Tur01} suggests that the signal detected in each receiver element may be sampled at a rate far below Nyquist, by modeling it as an FRI signal.  The authors propose that $\varphi_m\left(t\right)$, detected in the $m$th element, be regarded as sum of a relatively small number of pulses, all replicas of some known pulse shape.  Explicitly: 
\begin{equation}
\label{E:01}
\varphi_m\left(t\right)=\sum_{l=1}^L{a_{l,m}h\left(t-t_{l,m}\right)}.
\end{equation} 
Here $L$ is the number of scattering elements, distributed throughout the sector radiated by the transmitted pulse, $t_{l,m}$ denotes the time in which the reflection from the $l$th element arrived at the $m$th receiver, and $a_{l,m}$ denotes the reflection's amplitude, as detected by the $m$th receiver.  Finally, $h\left(t\right)$ denotes the known pulse shape, regarded, in our work, by the term two-way pulse.  The signal in \eqref{E:01} is completely defined by $2L$ real-valued parameters, $\left\{t_{l,m},a_{l,m}\right\}_{l=1}^L$.  

Sampling FRI signals was first treated by Vetterli et. al. \cite{Vetterli01}.  Their approach involves projecting the FRI signal, characterized by $2L$ degrees of freedom per unit time, onto a $2L$-dimensional subspace, corresponding to a subset of its Fourier series coefficients. Having extracted $2L$ frequency samples of the signal, spectral analysis techniques (e.g. annihilating filter~\cite{Stoica01}, matrix pencil~\cite{Tapan01}) may be applied, in order to extract the unknown signal parameters. Applying this solution to the problem formulated in \eqref{E:01}, \cite{Tur01} formalizes the relationship between the ultrasound signal's Fourier series coefficients to its unknown parameters, as a spectral analysis problem. 

Let $T$ be the duration of $\varphi_m\left(t\right)$.  We can then expand $\varphi_m\left(t\right)$ in a Fourier series, with coefficients
\begin{align}
\label{E:06}
\begin{split}
\phi_m\left[k\right]&=\frac{1}{T}\int_{0}^{T}{\varphi_m\left(t\right)e^{-i\frac{2\pi}{T}kt}dt}\\
&=\frac{1}{T}\int_{0}^{T}{\sum_{l=1}^L{a_{l,m}h\left(t-t_{l,m}\right)}e^{-i\frac{2\pi}{T}kt}dt}\\
&=\frac{1}{T}H\left(\frac{2\pi}{T}k\right)\sum_{l=1}^L{a_{l,m}e^{-i\frac{2\pi}{T}kt_{l,m}}},
\end{split}
\end{align} 
where $H\left(\omega\right)$ denotes the Continuous Time Fourier Transform (CTFT) of $h\left(t\right)$. Consider the sequence $\left\{k_{j,m}\right\}_{j=1}^{K_m}$, comprising $K_m$ integers, and define the length-$K_m$ vector $\bf{\Phi_m}$ with $j$th element $\phi_m\left[k_{j,m}\right]$.  Then \eqref{E:06} may be written in matrix form:
\begin{align}
\label{E:07}
\begin{split}
{\bf{\Phi_m}}=\frac{1}{T}\bf{H_m V_m a_m},
\end{split}
\end{align}
where ${\bf{H_m}}$ is a diagonal matrix with diagonal elements $H\left(\frac{2\pi}{T}k_{j,m}\right)$, ${\bf{V_m}}$ contains $e^{-i\frac{2\pi}{T}k_j t_{l,m}}$ as its $\left(j,l\right)$th element, and ${\bf{a_m}}$ is the length $L$ vector, with elements $a_{l,m}$. Choosing $k_{j,m}$ such that $H\left(\frac{2\pi}{T}k_{j,m}\right)\neq 0$,we can express \eqref{E:07} as: 
\begin{align}
\label{E:08}
\begin{split}
{\bf{y_m}}=\bf{V_m a_m},
\end{split}
\end{align} 
where ${\bf{y_m}}=T{\bf{H_m^{-1}}\Phi_m}$. If the values $k_{j,m}$ are a sequence of consecutive indices, then ${\bf{V_m}}$ takes on a Vandermonde form, and has full column rank~\cite{Stoica01} as long as $K_m\geq L$ and the time-delays are distinct, i.e., $t_{i,m} \neq t_{j,m}$, for all $ i\neq j$. The formulation derived in \eqref{E:08} is a standard spectral analysis problem.  As long as $K_m\geq 2L$, it may be solved for the unknown parameters $\left\{t_{l,m},a_{l,m}\right\}_{l=1}^L$, using methods such as annihilating filter~\cite{Stoica01} or matrix pencil~\cite{Tapan01}.  

Having obtained \eqref{E:07}, the sampling scheme reduces to the problem of extracting $K_m$ frequency samples of $\varphi_m\left(t\right)$, where $K_m\geq2L$.  A single-channel Xampling scheme, such as the one derived in \cite{Tur01}, allows robust estimation of such coefficients from point-wise samples of the signal, after filtering it with an appropriate kernel.  The estimation is performed by applying a linear transformation to $p$ complex-valued samples (equivalently, $2p$ real-valued samples) of the filtered signal, requiring that $p\geq K_m$.  In this context,~\cite{Tur01} introduces the Sum of Sincs kernel, which satisfies the necessary constraints, and is additionally characterized by a finite temporal support.  Combining the requirements that $K_m\geq2L$ and $p\geq K_m$, the Xampling scheme proposed in \cite{Tur01} allows reconstruction of the signal detected in each receiver element from a minimal number of $4L$ real-valued samples.  Considering the nominal figures derived in the previous section for standard beamforming, we conclude that, as long as $4L\ll1680$, such a Xampling method may indeed achieve a substantial rate reduction.  

\section{Why Compressed Beamforming?}
\label{sec:04}

Applied to a single receiver element, the Xampling scheme proposed in~\cite{Tur01} achieves good signal reconstruction for an actual ultrasound signal, reflected from a setup of phantom targets.  In principle, we could apply this approach to each receiver element individually, resulting in a parametric representation for each of the signals $\left\{\varphi_m\left(t\right)\right\}_{m=1}^M$.  Being able to digitally reconstruct the detected signals, we could then proceed with the standard beamforming process, outlined in Section \ref{sec:02}, aimed at constructing the corresponding scanline.  Computational effort would have been reduced, by limiting the beamforming process to the support of the estimated pulses.  In fact, we could possibly bypass the beamforming stage, by deriving a geometric model which maps the set of delays, $\left\{t_{l,m}\right\}_{m=1}^M$, associated with the $l$th reflector, to its two-dimensional position ${\bf{p_l}}=\left(x_l,z_l\right)$.  However, applying the proposed FRI Xampling scheme to signals reflected from biological tissues, we face two fundamental obstacles: low SNR and proper interpretation of the estimated signal parameters, considering the profile of the transmitted beam.  These two difficulties may be better understood by examining Fig.~\ref{Fig:02}, which depicts traces acquired for cardiac images of a healthy consenting volunteer using a GE breadboard ultrasonic scanner.  

\begin{figure*}
\begin{minipage}[b]{0.5\linewidth}
\centering
\centerline{\includegraphics[width=7.0cm]{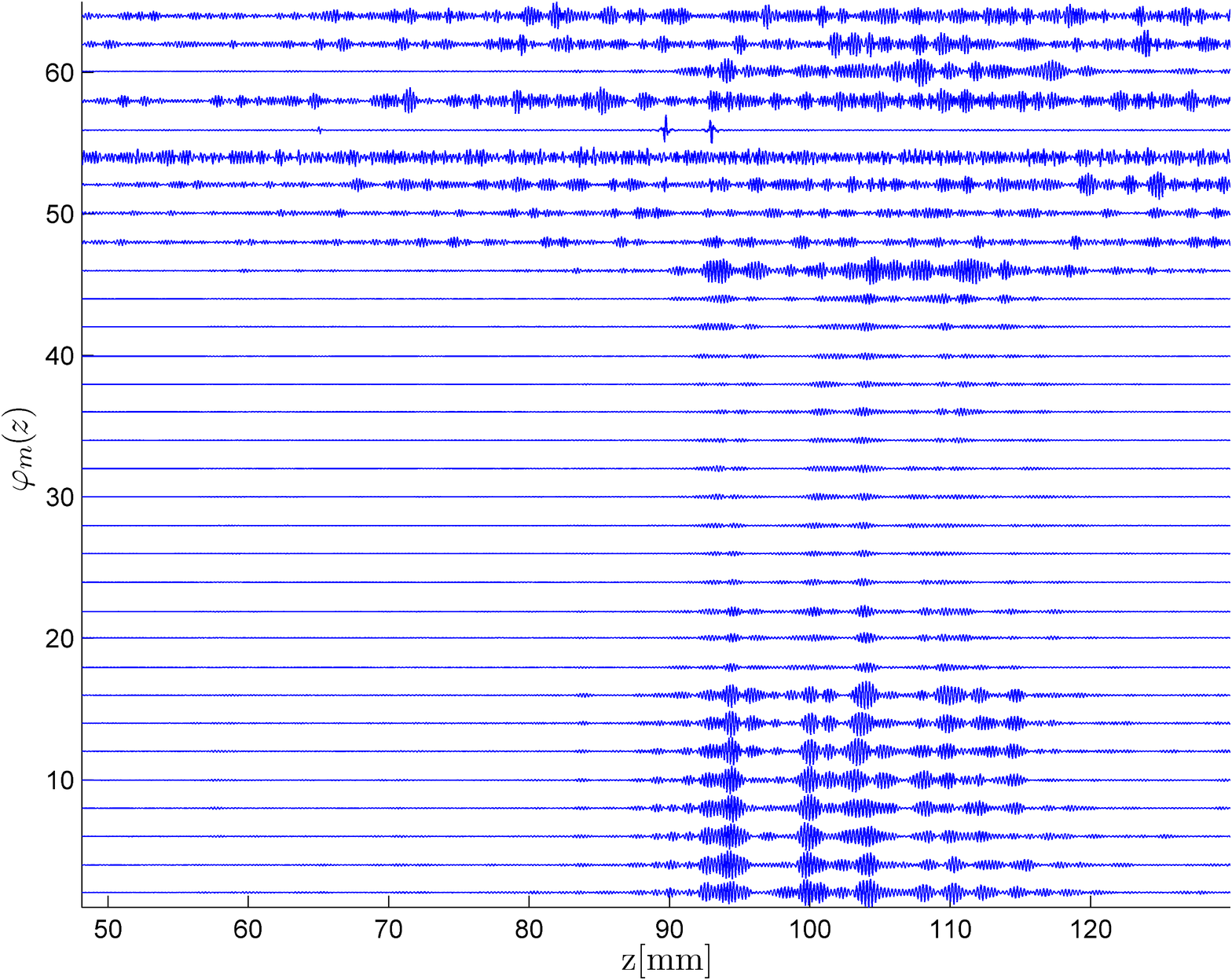}}
\centerline{(a)}
\end{minipage}
\begin{minipage}[b]{0.5\linewidth}
\centering
\centerline{\includegraphics[width=7.0cm]{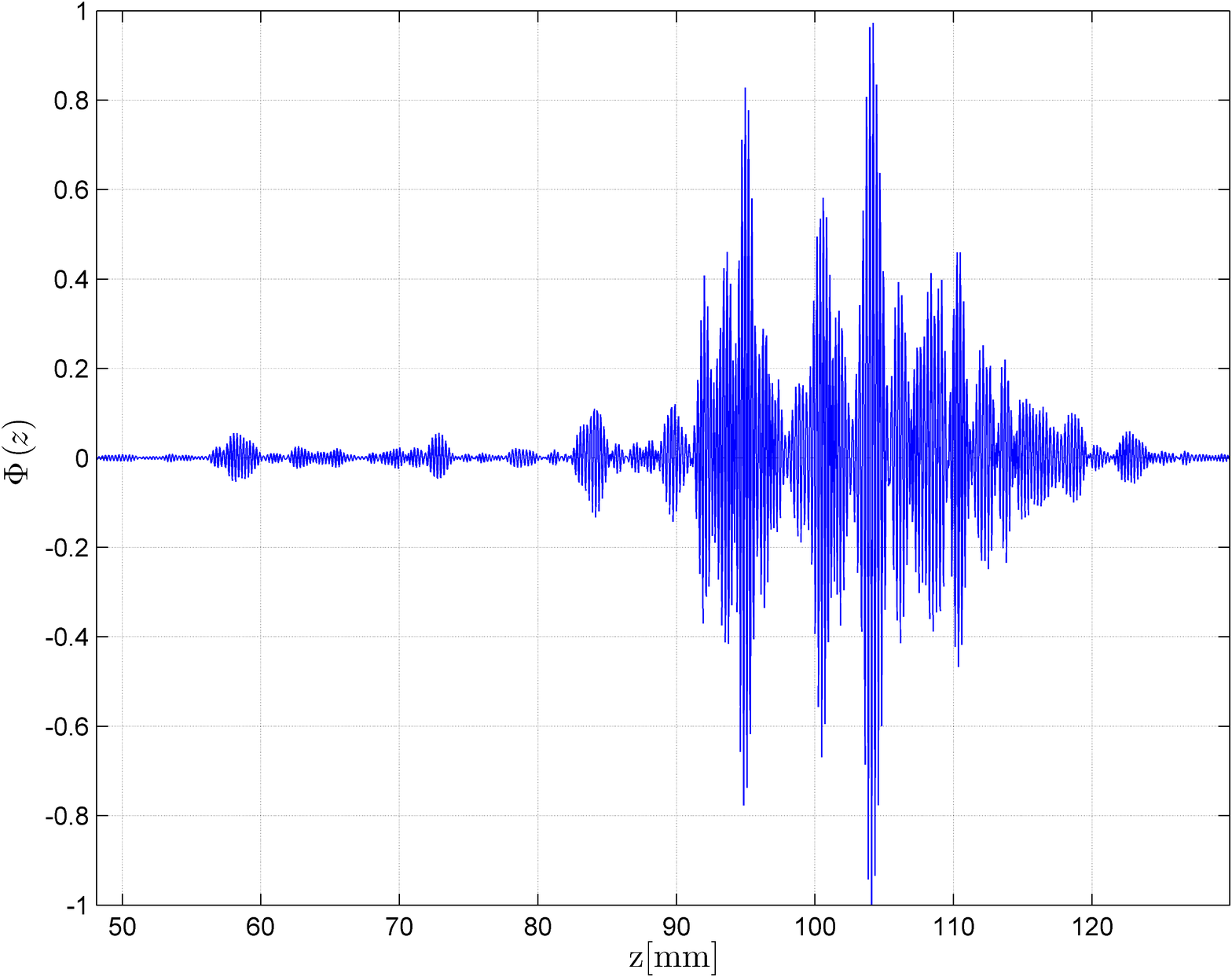}}
\centerline{(b)}
\end{minipage}
\caption{(a) Signals detected for cardiac images following the transmission of a single pulse.  The vertical alignment of each trace matches the index of the corresponding receiver element.  (b) Beamformed signal obtained by combining the detected signals with appropriate, time-varying time delays.  The data is  acquired using a GE breadboard ultrasonic scanner.}
\label{Fig:02}
\end{figure*}

In the left plot (a), are signals detected by 32 of 64 active array elements, following the transmission of a single pulse.  The pulse was conducted along a narrow beam, forming an arbitrary angle $\theta$ with the $\hat{z}$ axis.  The right plot (b) depicts the signal obtained by applying beamforming to the detected signals, as outlined in Section \ref{sec:02}.  Examining the individual traces, one notices the appearance of strong pulses, possibly overlapping, characterized by a typical shape, as proposed in \eqref{E:01}.  Let us assume that we could indeed extract the delays and 
amplitudes of these pulses, by applying the proposed FRI Xampling scheme to each element.  We suggested that beamforming could be bypassed, by deriving a geometric model for estimating the two-dimensional position of a scattering element, based on the delays of pulses associated with it, yet estimated in different receivers.  In order to apply such a model, we must first be able to match corresponding pulses across the detected signals.  However, referring to the practical case depicted in (a), we notice that such a task is not at all trivial - the individual signals depict reflections, originating from the entire sector, radiated by the transmitted pulse. These reflections may, therefore, vary significantly across traces.  In fact, some pulses, visible in several traces, are not at all apparent in other traces.  In contrast, the beamformed signal, by its construction, depicts intensity of reflections originating from along the central transmission axis, while attenuating reflections originating off this axis.

Attempting to apply FRI Xampling to each receiver element individually, we encounter an even more fundamental obstacle, at the earlier stage of extracting the signal's parametric representation from its low-rate samples. The individual traces contain high levels of noise. The noisy components, especially noticeable in traces 54 and 64, rise mainly from constructive and destructive interference of acoustic waves, reflected by dense, sub-wavelength scatterers in the tissue.  The latter are typically manifested as granular texture in the ultrasound image, called speckle, after a similar effect in laser optics~\cite{Szabo01}. The noisy components inherently induce erroneous results, when attempting to sample and reconstruct the FRI components using the Xampling approach.  In extreme scenarios, where the noise masks the FRI component, the extracted parameters will be meaningless, such that any attempt to cope with errors in the parametric domain will turn out useless.   

The motivation to our approach rises from the observation, that we may resolve the aforementioned obstacles by Xampling the beamformed signal, $\Phi\left(t;\theta\right)$, rather than the individual signals $\varphi_m\left(t\right)$.  Whereas beamforming is a fundamental process in ultrasound imaging since its early days, our innovation regards its integration into the Xampling process.  We derive our compressed beamforming approach, beginning with conceptual Xampling of the beamformed signal, using the scheme proposed in \cite{Gedalyahu01}.  We then show that an equivalent result may be obtained from low-rate samples of the individual signals  $\varphi_m\left(t\right)$.   

A necessary condition for implementing our approach is that $\Phi\left(t;\theta\right)$, generated from $\left\{\varphi_m\left(t\right)\right\}_{m=1}^M$ satisfying \eqref{E:01}, is also FRI of similar form.  Examining Fig.~\ref{Fig:02} we notice that $\Phi\left(t;\theta\right)$ exhibits a structure similar to that of the individual signals, comprising strong pulses of typical shape, which may overlap. In this case, there are several obvious advantages in Xampling $\Phi\left(t;\theta\right)$.  First, since $\left\{\varphi_m\left(t\right)\right\}_{m=1}^M$  are averaged in $\Phi\left(t;\theta\right)$ (after appropriate distortion, derived from the acoustic reciprocity theorem) it naturally exhibits enhanced SNR with respect to the individual signals.  The attenuation of noise in the beamformed signal, compared to the individual signals, is apparent in Fig.~\ref{Fig:02}, especially in the interval $50\mbox{mm}-80\mbox{mm}$.    Second, $\Phi\left(t;\theta\right)$ is directly related to an individual scanline. This means that we are no longer bothered with the ambiguous problem of matching pulses across signals detected in different elements.  Finally, recall that the signal model derived in \eqref{E:01} assumes isolated point-reflectors.  Such a model is better justified with respect to $\Phi\left(t;\theta\right)$ since, by narrowing the effective width of the imaging beam, we may indeed approximate its intersection with reflecting structures to be point-like.  This effect is noticeable in Fig.~\ref{Fig:02} where some pulses, visible in individual traces, appear attenuated in the beamformed signal.  Such pulses correspond to reflectors located off the central axis of the transmission beam.  

In the next section, we focus on justifying the assumption that $\Phi\left(t;\theta\right)$ may be treated within the FRI framework.  An additional challenge, implied in Section \ref{sec:02}, regards the fact that $\Phi\left(t;\theta\right)$ does not exist in the analog domain - standard ultrasound devices generate it digitally, from samples of the signals detected in multiple receiver elements, taken at the Nyquist-rate.  Our goal is, therefore, to derive a scheme, which manages to estimate the necessary samples of $\Phi\left(t;\theta\right)$, from low-rate samples of filtered versions of $\left\{\varphi_m\left(t\right)\right\}_{m=1}^M$.


\section{Compressed Beamforming}
\label{sec:05}
Our approach is based on the assumption that the FRI scheme, outlined in Section \ref{sec:03}, may be applied to the beamformed signal $\Phi\left(t;\theta\right)$, constructed according to \eqref{E:04}-\eqref{E:05}.  The latter exhibits much better SNR than signals detected in individual receiver elements.  Additionally, it depicts reflections originating from a sector much narrower than the one radiated by the transmission beam.  Its translation into a single scanline is therefore straightforward.  In Section \ref{sec:51} we prove that if the signals $\varphi_m\left(t\right)$ obey the FRI model \eqref{E:01}, then $\Phi\left(t;\theta\right)$ is {\bf{approximately}} of the form: 
\begin{equation}
\label{E:09}
\Phi\left(t;\theta\right)=\sum_{l=1}^L{b_l h\left(t-t_{l}\right)},
\end{equation}
where $t_l$ denotes the time in which the reflection from the $l$th element arrived at the reference receiver, indexed $m_0$.  $\Phi\left(t;\theta\right)$ may thus be sampled using the Xampling schemes derived in \cite{Tur01,Gedalyahu01}.  In practice, we cannot sample $\Phi\left(t;\theta\right)$ directly, since it does not exist in the analog domain.  In Second \ref{sec:52} we show how the desired low-rate samples of $\Phi\left(t;\theta\right)$ can be determined from samples of $\varphi_m\left(t\right)$.    

\subsection{FRI Modeling of the Beamformed Signal}
\label{sec:51}
Throughout this section we apply three reasonable assumptions.  First, we assume that $2\gamma_m\leq t_l$.  Practically, such a constraint may be forced by appropriate apodization, as often performed in ultrasound imaging.  Namely, $\varphi_m\left(t\right)$ is combined in $\Phi\left(t;\theta\right)$ only for $t\geq 2\gamma_m$.  As an example, for the breadboard ultrasonic scanner used in our experiments, the array comprised $64$ receiver elements, distanced $0.29\mbox{mm}$ apart. The proposed apodization implies that the receivers located farthest from the origin are combined in the beamformed signal for imaging depth greater than $9.1\mbox{mm}$.  Second, we assume the two-way pulse, $h(t)$, to be compactly supported on the interval $\left[0,\Delta\right)$.  Finally, we assume that $\Delta \ll t_l$.  The last assumption may also be forced by appropriate apodization.  As an example, the nominal duration of the pulse acquired by the breadboard ultrasonic scanner used in our experiments was 
 $4\mu \mbox{sec}$.  In this case, echoes scattered from depth greater than $3.1\mbox{cm}$ already satisfy  $t_l>10\Delta$.     

Suppose that $\varphi_m\left(t\right)$ can be written as in \eqref{E:01}.  Applying the beamforming distortion \eqref{E:04}, we get
\begin{equation}
\label{E:10}
\begin{split}
{\hat{\varphi}}_m(t;\theta)=\sum_{l=1}^L{a_{l,m}h\left(\tau_m\left(t;\theta\right)-t_{l,m}\right)}.
\end{split}
\end{equation}
The resulting signal comprises $L$ pulses, which are distorted versions of the two-way pulse $h\left(t\right)$.  Suppose that some of the pulses originated in reflectors located off the central beam axis.  Beamforming implies that, once averaging the distorted signals according to \eqref{E:05}, such pulses will be attenuated due to destructive interference.  Being interested in the structure of the beamformed signal $\Phi\left(t;\theta\right)$, we are therefore concerned only with pulses which originated in reflectors located along the central beam.  For convenience, we assume that all pulses in \eqref{E:10} satisfy this property (pulses which do not satisfy it, will vanish in $\Phi\left(t;\theta\right)$).  We may thus use $\tau_m\left(t;\theta\right)$, defined in \eqref{E:04}, in order to express $t_{l,m}$ in terms of $t_l$.  Substituting $t=t_l$ into $\tau_m\left(t;\theta\right)$, we get $t_{l,m}=\tau_m\left(t_l;\theta\right)$, so that \eqref{E:10} becomes
\begin{equation}
\label{E:12}
\begin{split}
{\hat{\varphi}}_m(t;\theta)&=\sum_{l=1}^L{a_{l,m}\tilde{h}_{l,m}\left(t;\theta\right)},
\end{split}
\end{equation}
where we defined $\tilde{h}_{l,m}\left(t;\theta\right)=h\left(\tau_m\left(t;\theta\right)-\tau_m\left(t_l;\theta\right)\right)$. 

Applying our second assumption, the support of $\tilde{h}_{l,m}\left(t;\theta\right)$ is defined by the requirement that
\begin{align}\label{E:13}
0\leq\tau_m\left(t;\theta\right)-\tau_m\left(t_l;\theta\right)<\Delta.
\end{align}   
Using \eqref{E:13} and \eqref{E:04}, it is readily seen that $\tilde{h}_{l,m}\left(t;\theta\right)$ is supported on $\left[t_l,t_l+\Delta'\right)$, where
 \begin{align}\label{E:14}
\Delta'=2\Delta\frac{\sqrt{t_l^2-4\gamma_m t_l \sin\theta + 4\gamma_m^2}+\Delta}{\sqrt{t_l^2-4\gamma_m t_l \sin\theta + 4\gamma_m^2}+2\Delta+t_l-2\gamma_m\sin\theta}.
\end{align}   

\noindent Further applying our assumption that $2\gamma_m\leq t_l$, we obtain $\Delta'\leq2\Delta$.  

We have thus proven that $\tilde{h}_{l,m}\left(t;\theta\right)=0$ for $t\notin\left[t_l, t_l+2\Delta\right)$.  Next, let us write any $t$ in $\left[t_l, t_l+2\Delta\right)$ as $t=t_l+\eta$, where $0\leq \eta < 2\Delta$. Then 
\begin{align}\label{E:16}
\begin{split}
\tilde{h}_{l,m}\left(t;\theta\right)&=h\left(\tau_m\left({t_l+\eta};\theta\right)-\tau_m\left({t_l};\theta\right)\right).
\end{split}
\end{align}

\noindent We now rely on our assumption that $\Delta \ll t_l$.  Since $\eta<2\Delta$, we also have $\eta\ll t_l$.  The argument of $h\left(\cdot\right)$ in \eqref{E:16} may therefore be approximated, to first order, as  
\begin{align}\label{E:17}
&\tau_m\left({t_l+\eta};\theta\right)-\tau_m\left({t_l};\theta\right)=\sigma_{m,l}\left(\theta\right)\eta+o\left(\eta^2\right),   
\end{align}
where 
\begin{align}\label{E:15}
\sigma_{m,l}\left(\theta\right)= \frac{1}{2}\left(1+\frac{t_l-2\gamma_m\sin\theta}{\sqrt{t_l^2-4\gamma_mt_l\sin\theta+4\gamma_m^2}}\right).
\end{align}

\noindent Up until this point, we assumed that $2\gamma_m\leq t_l$.  Further assuming that $\gamma_m\ll t_l$, $\sigma_{m,l}\left(\theta\right)\rightarrow 1$.  Replacing $\eta$ by $\eta=t-t_l$, \eqref{E:16} may therefore be written as
\begin{align}\label{E:19}
\begin{split}
\begin{array}{ll}\tilde{h}_{l,m}\left(t;\theta\right)\approx h\left(t-t_l\right)&t\in\left[t_l,t_l+2\Delta\right).\end{array}
\end{split}
\end{align}

Combining \eqref{E:19} with the fact that $h\left(t-t_l\right)$ is zero outside $\left[t_l,t_l+2\Delta\right)$, \eqref{E:12} may be approximated as
 \begin{align}\label{E:20}
{\hat{\varphi}}_m(t;\theta) \approx \sum_{l=1}^L{a_{l,m}h\left(t-t_l\right)}.
\end{align} 
Averaging the signals $\left\{\hat{\varphi}_m(t;\theta)\right\}_{m=1}^M$ according to \eqref{E:05}, we get:
\begin{equation}
\label{E:21}
\begin{split}
\Phi\left(t;\theta\right)&\approx\sum_{l=1}^L{\left(\frac{1}{M}\sum_{m=1}^M{a_{l,m}}\right)h\left(t-t_l\right)}=\sum_{l=1}^L{b_l h\left(t-t_l\right)},
\end{split}
\end{equation}   
which is indeed the FRI form \eqref{E:09}.  Additionally, assuming that the support of $\varphi_m\left(t\right)$ is contained in $\left[0,T\right)$, we show in the Appendix that there exists $T_B\left(\theta\right)\leq T$, such that the support of $\Phi\left(t;\theta\right)$ is contained in $\left[0,T_B\left(\theta\right)\right)$ and, additionally, $\tau_m\left({T_B\left(\theta\right)};\theta\right)\leq T$.  

As $\gamma_m$ grows towards $t_l$, $\sigma_{m,l}\left(\theta\right)$ decreases, resulting in a larger distortion of the $l$th pulse.  Consequently, the approximation of ${\hat{\varphi}}_m(t;\theta)$ as a sum of shifted replicas of the two-way pulse becomes less accurate.  The Xampling schemes used by \cite{Tur01,Gedalyahu01} rely on the projection of the detected signal onto a subspace of its Fourier series coefficients.  We therefore examine the dependency of the projection error on the distortion parameters, $\gamma_m$, $t_l$ and $\theta$.  In Fig. \ref{Fig:03}, we show projection errors calculated numerically, for a signal comprising a single pulse of duration $\Delta=2\mu\mbox{sec}$.  The pulse was simulated by modulating a Gaussian envelope with carrier frequency $3\mbox{MHz}$.  It was then shifted by multiple time delays, $t_l$, where $0\leq t_l \leq T$, and $T=210\mu\mbox{sec}$, corresponding to an imaging depth of $16\mbox{cm}$.  For each delay, we generated the signals $\varphi_m\left(t\right)$, assuming that the reflector is positioned along the $\hat{z}$ axis ($\theta=0$), and that the receiver elements are distributed $0.29\mbox{mm}$ apart, along the $\hat{x}$ axis.  We chose $M=63$, such that the center (reference) receiver was indexed $m_0=32$.  The beamforming distortion was then applied to the simulated signals, based on \eqref{E:04}.  Finally, the distorted signals were projected onto a subset of $K=121$ consecutive Fourier series coefficients, taken within the essential spectrum of the two-way pulse.  The coefficients extracted from the $m$th distorted signal were arranged into the length $K$ vector, ${\bf{\Phi_m}}$.  As implied by \eqref{E:04}, no distortion is applied to the signal detected at the reference receiver.  We therefore evaluate the projection error by calculating the SNR defined as $20\log_{10}{\frac{\|{\bf{\Phi_{m_0}}}\|_2}{\|{\bf{\Phi_m}}-{\bf{\Phi_{m_0}}}\|_2}}$.  

The traces obtained for several values of $1\leq m<32$ are depicted in the figure.  As $t_l$ grows, $\sigma_{m,l}\left(\theta\right)$ approaches 1, and the approximation \eqref{E:20} becomes more valid.  As a result, the projection error decreases.  For receivers located near the origin, such that $\delta_m\ll10\mbox{mm}$, the error decreases very quickly.  For instance, examining $\delta_{31}=0.29\mbox{mm}$, the SNR grows above $25\mbox{dB}$ for a reflection originating at distance greater than $1/50$ of the imaging depth.  The SNR improves more moderately for receivers located farther away from the origin.  Nevertheless, considering the receiver located farthest away from the origin, $\delta_{1}=8.99\mbox{mm}$, the SNR grows above $10\mbox{dB}$ for a reflection originating at distance greater than $1/5$ of the imaging depth.   

\begin{figure}
\centerline{\includegraphics[width=7.50cm]{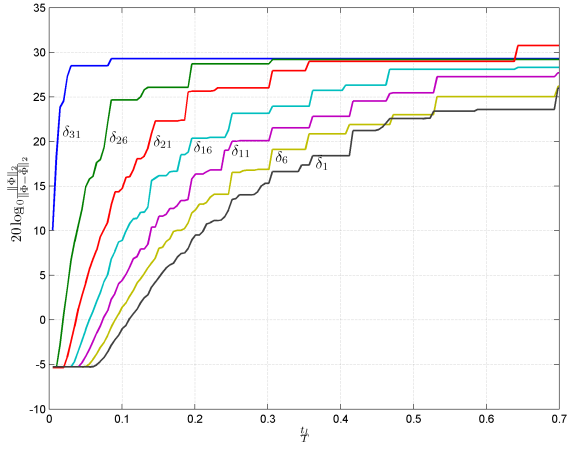}}
\caption{Projection error caused by beamforming distortion with $\theta=0$ vs. pulse delay, $t_l$, for several receiver elements.  The elements are distributed $0.29\mbox{mm}$ apart, such that $\delta_1=8.99\mbox{mm}$ (element farthest from array center) and  $\delta_{31}=0.29\mbox{mm}$.  Zero error is obtained for the center element, $\delta_{32}$, since no distortion is required in this case.}
\label{Fig:03}
\end{figure}


Concluding this section, our empirical results indeed justify the approximation proposed in \eqref{E:09}, where appropriate apodization may further improve this approximation.  Assuming \eqref{E:09} to be valid, we may reconstruct the beamformed signal using the Xampling schemes proposed in \cite{Tur01,Gedalyahu01}.  

\subsection{Compressed Beamforming with Distorted Analog Kernels}
\label{sec:52}
An obvious problem is that $\Phi\left(t;\theta\right)$ does not exist in the analog domain, and therefore may not be Xampled directly. We now propose a modified Xampling scheme, which allows extraction of its necessary low-rate samples, by sampling filtered versions of $\varphi_m\left(t\right)$ at sub-Nyquist rates.    

Since the support of $\Phi\left(t;\theta\right)$ is contained in $\left[0,T_B\left(\theta\right)\right)$, where $T_B\left(\theta\right)\leq T$, we may define $\Phi\left(t;\theta\right)$'s Fourier series with respect to the interval $\left[0,T\right)$.  Denoting by $c_j$ the $k_j$th Fourier series coefficient of $\Phi\left(t;\theta \right)$, we have 

\begin{align}
\begin{split}
\label{E:23}
c_j&=\frac{1}{T}\int_0^{T}{I_{\left[0,T_B\left(\theta\right)\right)}\left(t\right)\Phi\left(t;\theta \right)e^{-i\frac{2\pi}{T}k_j t}dt}, 
\end{split}
\end{align}
where $I_{\left[a,b\right)}\left(t\right)$ is the indicator function, taking the value 1 for $a\leq t<b$ and 0 otherwise.  Plugging the indicator function in \eqref{E:23} may seem unnecessary.  However, once transforming \eqref{E:23} into an operator applied directly to $\left\{\varphi_m\left(t\right)\right\}_{m=1}^M$, it serves an important role in zeroing intervals, which are assumed zero according to \eqref{E:01}, but, in any practical implementation, contain noise.  Substituting \eqref{E:05} into \eqref{E:23}, we can write      
\begin{align}
\begin{split}
\label{E:24}
c_j=\frac{1}{M}\sum_{m=1}^{M}{c_{j,m}}, 
\end{split}
\end{align}
where, from \eqref{E:04},
\begin{align}
\begin{split}
\label{E:25}
c_{j,m}&=\frac{1}{T}\int_0^{T}{I_{\left[0,T_B\left(\theta\right)\right)}\left(t\right) \varphi_m\left(\tau_m\left({t};\theta\right)\right) e^{-i\frac{2\pi}{T}k_j t}dt}\\ 
&=\frac{1}{T}{\int_{0}^{T}g_{j,m}(t;\theta)\varphi_m\left(t\right)dt},
\end{split}
\end{align}
and
\begin{align}\label{E:27}    
\begin{split}
g_{j,m}(t;\theta)=&q_{j,m}(t;\theta)e^{-i\frac{2\pi}{T}k_j t},\\ 
q_{j,m}(t;\theta)=&I_{\left[|\gamma_m|,T_m\left(\theta\right)\right)}\left(t\right)\left(1+\frac{\gamma_m^2  \cos^2\theta}{\left(t-\gamma_m \sin\theta\right)^2}\right) \times \\
&\exp\left\{i \frac{2\pi}{T} k_j\frac{\gamma_m-t\sin\theta}{t-\gamma_m\sin\theta}\gamma_m \right\},\\
T_m\left(\theta \right)=&\tau_m\left({T_B\left(\theta\right)};\theta\right).
\end{split}
\end{align}
 
The process defined in \eqref{E:24}-\eqref{E:27} can be translated into a multi-channel Xampling scheme, such as the one depicted in Fig. \ref{Fig:04}.  Each signal $\varphi_m\left(t\right)$ is multiplied by a bank of kernels $\left\{g_{j,m}\left(t;\theta\right)\right\}_{j=1}^K$ defined by \eqref{E:27}, and integrated over $\left[0, T\right)$.  This results in a vector ${\bf{c_m}}=\left[\begin{array}{cccc}c_{1,m}&c_{2,m}&...&c_{K,m}\end{array}\right]^T$.  The vectors $\left\{{\bf{c_m}}\right\}_{m=1}^M$ are then averaged in ${\bf{c}}=\left[\begin{array}{cccc}c_{1}&c_{2}&...&c_{K}\end{array}\right]^T$, which has the desired improved SNR property, and provides a basis for extracting the $2L$ parameters which define $\Phi\left(t;\theta\right)$.  Since $\Phi\left(t;\theta\right)$ satisfies \eqref{E:09}, we apply a similar derivation to that outlined in Section~\ref{sec:04}, yielding 
\begin{align}\label{E:47}    
\begin{split}
{\bf{c}}=\frac{1}{T}{\bf{HVb}},
\end{split}
\end{align}
where ${\bf{H}}$ is a diagonal matrix with $j$th diagonal element $H\left(\frac{2\pi}{T}k_{j}\right)$, ${\bf{V}}$ contains $e^{-i\frac{2\pi}{T}k_j t_{l}}$ as its $\left(j,l\right)$th element, and ${\bf{b}}$ is the length $L$ vector, with elements $b_{l}$.  The matrix ${\bf{V}}$ may be estimated by applying spectral analysis techniques, allowing for the vector of coefficients ${\bf{b}}$ to be solved by a least squares approach~\cite{Tapan01}.  Fig. \ref{Fig:11} illustrates the shape of the resulting kernels $g_{j,m}\left(t;\theta\right)$, setting $\theta=0$ and choosing two arbitrary values of $k_j$.  For each choice of $k_j$ we plot the kernels corresponding to 7 receiver elements, selected from an array comprising 64 elements, distanced $0.49\mbox{mm}$ apart.        

\begin{figure}
\begin{minipage}[b]{1.0\linewidth}
\centering
\centerline{\includegraphics[width=8.5cm]{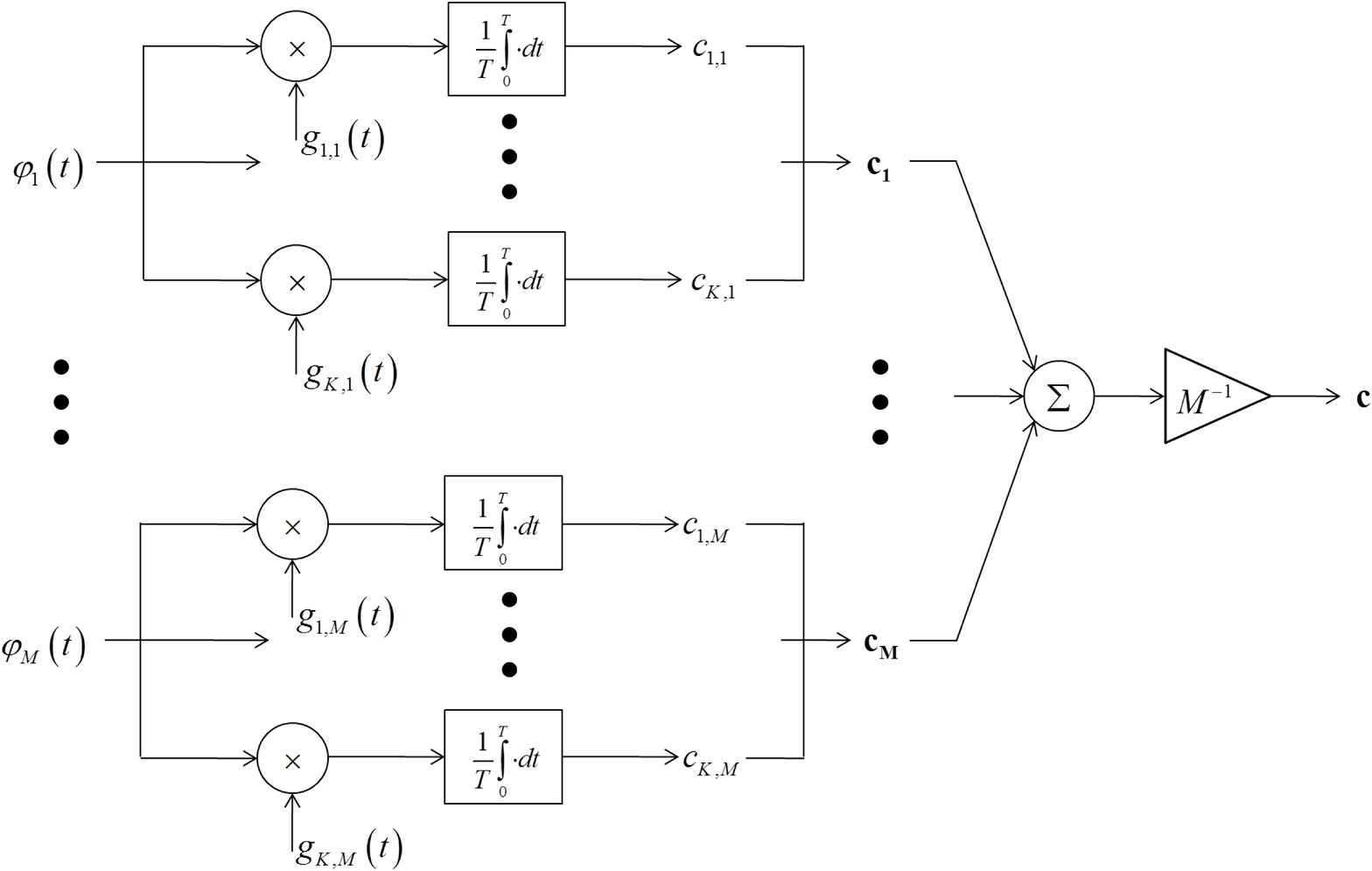}}
\end{minipage}
\caption{Xampling scheme utilizing distorted exponential kernels.}
\label{Fig:04}
\end{figure}     

\begin{figure}
\begin{minipage}[b]{1.0\linewidth}
\centering
\centerline{\includegraphics[width=7cm]{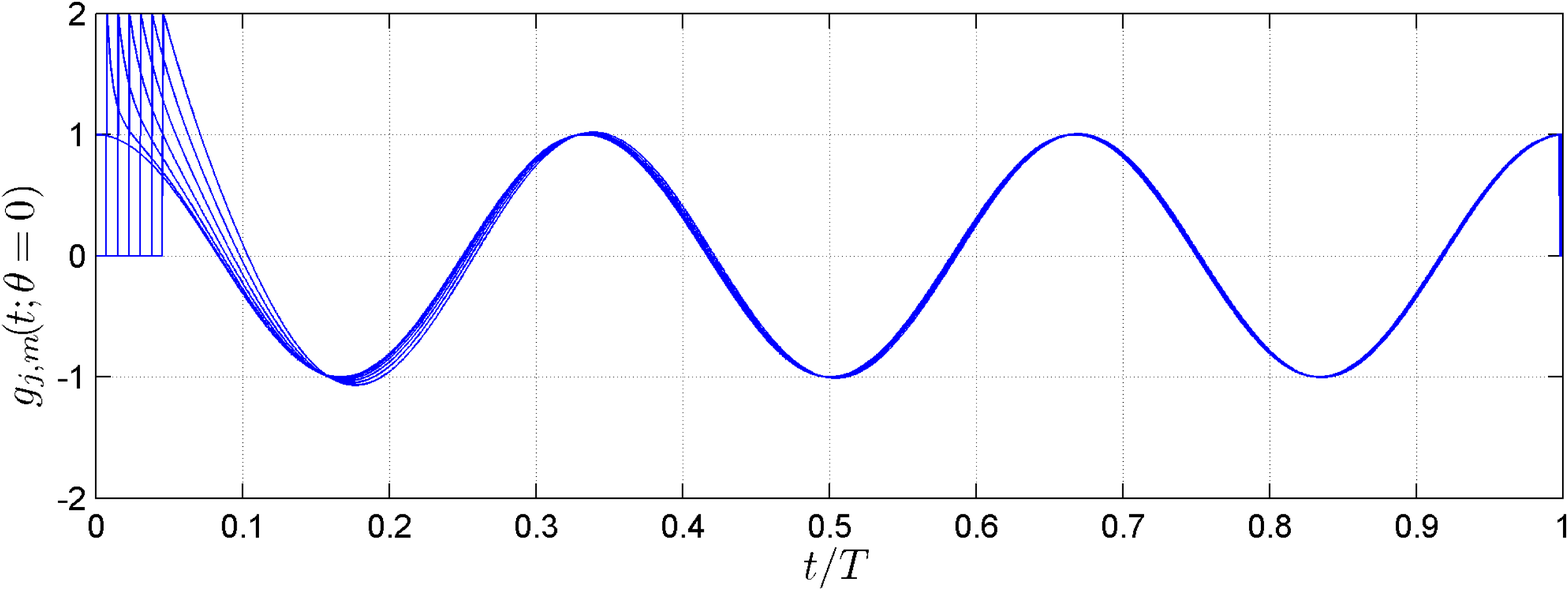}}
\centerline{(a)}
\end{minipage}
\begin{minipage}[b]{1.0\linewidth}
\centering
\centerline{\includegraphics[width=7cm]{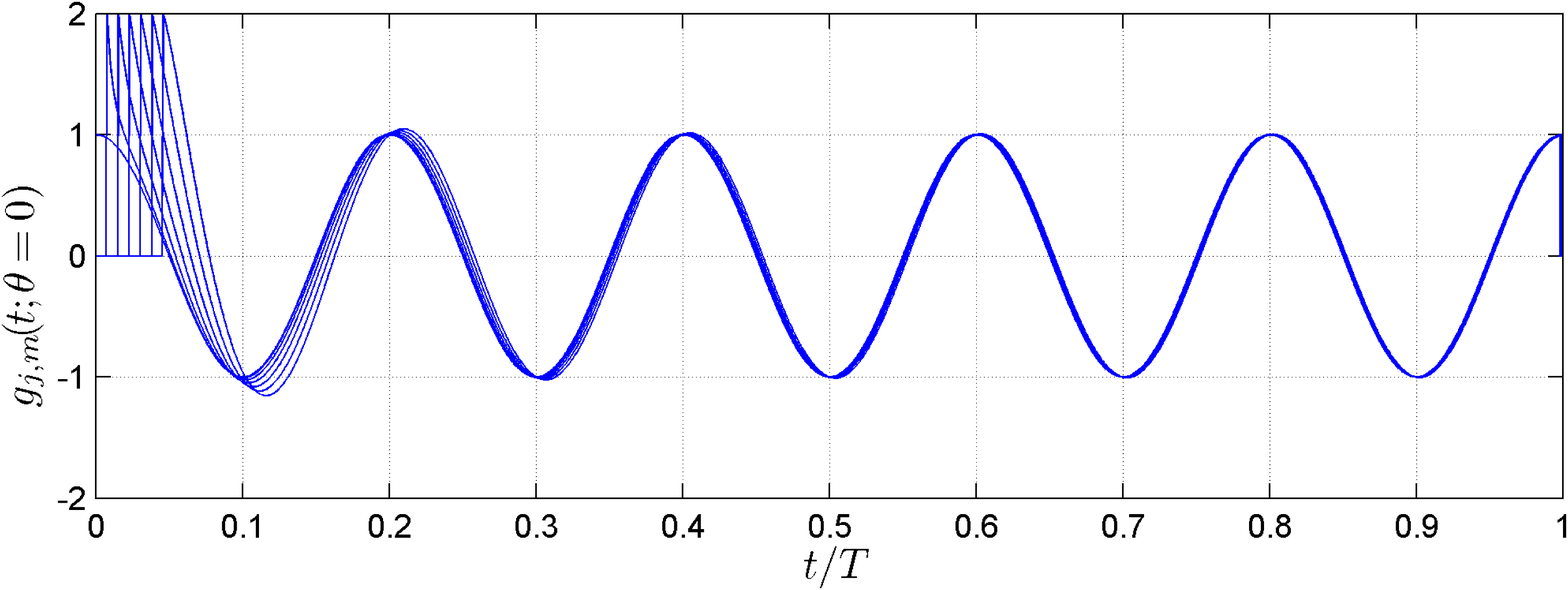}}
\centerline{(b)}
\end{minipage}
\caption{Real part of $g_{j,m}(t;\theta=0)$ for $T=210\mu\mbox{sec}$ and $k_j$ satisfying: (a) $k_j=3$, (b) $k_j$=5.  We assume an array comprising $M=64$ elements, distanced $0.49\mbox{mm}$ apart, and plot 7 traces which were obtained for the elements indexed $\left\{m_0, m_0+5, m_0+10, ..., m_0+30\right\}$.}
\label{Fig:11}
\end{figure}      
\section{Simplified Xampling Mechanism}
\label{sec:06}
In the previous section, we developed a Xampling approach to extract the Fourier series coefficients of $\Phi\left(t;\theta\right)$. However, the complexity of the resulting analog kernels, together with their dependency on $\theta$, makes hardware implementation of the scheme depicted in Fig. \ref{Fig:04} complex.  Here, we take an additional step, which allows the approximation of $\left\{c_{j,m}\right\}_{j=1}^K$, and consequently $\left\{c_{j}\right\}_{j=1}^K$, from low-rate samples of $\varphi_m\left(t\right)$, obtained in a much more straightforward manner.

We begin by substituting $\varphi_m\left(t\right)$ of \eqref{E:25} by its Fourier series, calculated with respect to $\left[0,T\right)$.  Denoting the $n$th Fourier coefficient by $\phi_m\left[n\right]$, we get: 
\begin{align}\label{E:28}    
\begin{split}
c_{j,m}&=\sum_{n}{\phi_m\left[n\right]\frac{1}{T}\int_0^T{q_{j,m}(t;\theta)e^{-i\frac{2\pi}{T}\left(k_j-n\right) t}dt}}\\
&=\sum_{n}{\phi_m\left[k_j-n\right]Q_{j,m;\theta}\left[n\right]},
\end{split}
\end{align}
where $Q_{j,m;\theta}\left[n\right]$ are the Fourier series coefficients of $q_{j,m}(t;\theta)$, also defined on $\left[0, T\right)$. Let us replace the infinite summation of \eqref{E:28} by its finite approximation:
\begin{align}\label{E:29}
\begin{split}
{\hat{c}}_{j,m}=\sum_{n=N_1}^{N_2}{\phi_m[k_j-n]Q_{j,m;\theta}\left[n\right]}.
\end{split}
\end{align} 
The following proposition shows that this approximation can be made sufficiently tight. 
\begin{mydef}
Assume that $\int_{-\infty}^{\infty}{\left|\varphi_m\left(t\right)\right|^2dt}<\infty$.  Then, for any $\epsilon>0$, and for any selection $\left(j,m;\theta\right)$, there exist finite $N_1\left(\epsilon, k_j, m; \theta \right)$ and $N_2\left(\epsilon, k_j, m;\theta \right)$ such that $|{c_{j,m}-\hat{c}}_{j,m}|^2<\epsilon$.
\end{mydef}

\begin{proof}
\noindent Let $l_2$ be the space of square-summable sequences, with norm $\|{\bf{x}}\|_2^2=\sum_n{|x_n|^2}$.  Let ${\bf{a}}=\left\{\phi_m\left[k_j-n\right]\right\}_{n=-\infty}^{\infty}$ and ${\bf{b}}=\left\{Q_{j,m;\theta}^*\left[n\right]\right\}_{n=-\infty}^{\infty}$.   Since $\varphi_m\left(t\right)$ is of finite energy, ${\bf{a}}\in l_2$.  We may calculate the $l_2$ norm of ${\bf{b}}$, based on the definition of $q_{j,m}\left(t;\theta\right)$ in \eqref{E:27}, resulting in $\|{\bf{b}}\|_2\approx T_m\left(\theta\right)/T<\infty$.  This implies that ${\bf{b}}\in l_2$ as well.  Let ${\bf{b_t}}$ be the truncated sequence ${\bf{b}}$ for $N_1\leq n\leq N_2$ and zero otherwise.  We may then write the approximation error as:
\begin{align}\label{E:30}
\begin{split}
|c_{j,m}-{\hat{c}}_{j,m}|^2=|\left<{\bf{a}},{\bf{b}-{\bf{b_t}}}\right>|^2\leq \|{\bf{a}}\|_{2}^2\|{\bf{b}}-{\bf{b_t}}\|_{2}^2,
\end{split}
\end{align}  
where $\left<\cdot,\cdot\right>$ is the inner product defined as $\left<{\bf{x}},{\bf{y}}\right>=\sum_n{x_ny_n^*}$.  The last transition in \eqref{E:30} is a result of Cauchy-Schwartz inequality.  By definition of ${\bf{b_t}}$ and ${\bf{b}}$, it is readily seen that $\|{\bf{b}}-{\bf{b_t}}\|_{2}^2={\|{\bf{b}}\|_{2}^2-\|{\bf{b_t}}\|_{2}^2}$.  Denoting $\rho^2=\|{\bf{b_t}}\|_2^2/\|{\bf{b}}\|_2^2$, \eqref{E:30} becomes 
\begin{align}\label{E:31}
\begin{split}
|c_{j,m}-{\hat{c}}_{j,m}|^2\leq \|{\bf{a}}\|_{2}^2\|{\bf{b}}\|_{2}^2\left(1-\rho^2\right).
\end{split}
\end{align}  

Since $\|{\bf{b}}\|_2<\infty$, $\rho^2$ can approach $1$ as close as we desire, by appropriate selection of $N_1$ and $N_2$.  For any $\epsilon>0$, there exists $\rho^2\left(\epsilon\right)< 1$, such that the right side of \eqref{E:31} is smaller than $\epsilon$.  Selecting $N_1$ and $N_2$ for which $\|{\bf{b_t}}\|_2^2/\|{\bf{b}}\|_2^2\geq \rho^2\left(\epsilon\right)$, results in $|c_{j,m}-{\hat{c}}_{j,m}|^2<\epsilon$, as required. Furthermore, setting an upper bound on the energy of $\varphi_m\left(t\right)$, and thereby on $\|{\bf{a}}\|_2^2$, $N_1$ and $N_2$ may be chosen off-line, subject to the decay properties of the sequence $\left\{Q_{j,m;\theta}\left[n\right]\right\}_{n=-\infty}^{\infty}$.
\end{proof}  

\noindent Using Proposition 1, we can compute ${\hat{c}}_{j,m}$ as a good approximation to $c_{j,m}$.  We now show how $\hat{c}_{j,m}$ can be obtained directly from the Fourier series coefficients $\phi_m\left[n\right]$ of each $\varphi_m\left(t\right)$.   

We first evaluate $N_1$ and $N_2$ for a certain choice of $m$ and $\theta$, such that $c_{j,m}$ may be approximated to the desired accuracy using \eqref{E:29}.  Equivalently, we obtain the minimal subset of $\varphi_m\left(t\right)$'s Fourier series coefficients, required for the approximation of $c_{j,m}$.  Performing this for all $1\leq j\leq K$, we obtain $K$ such subsets.  Denoting the union of these subsets by $\kappa_m$, we may now simultaneously compute  $\left\{\hat{c}_{j,m}\right\}_{j=1}^K$ from $\left\{\phi_m\left[n\right]\right\}_{n\in\kappa_m}$ by a linear transformation. Define the length-$K_m$ vector ${\bf{\Phi_m}}$, with $l$th element
 $\phi_m\left[k_l\right]$, and $k_l$ being the $l$th element in $\kappa_m$.  Using \eqref{E:29}, we may write
\begin{align}\label{E:32}
\begin{split}
{\bf{{\hat{c}}_m}}={\bf{A_{m}}}\left(\theta\right){\bf{\Phi_m}},
\end{split}
\end{align}  
where ${\bf{\hat{c}_m}}$ is the length-$K$ vector with $j$th element $\hat{c}_{j,m}$, and ${\bf{A_{m}}}\left(\theta\right)$ is a $K\times K_m$ matrix with elements
\begin{align}\label{E:33}
\begin{split}
a_{j,l}=\left\{\begin{array}{ll} Q_{j,m;\theta}\left[k_j-k_l\right] & N_1\left(k_j\right)  \leq k_j-k_l \leq N_2\left(k_j\right) \\ 0 & \mbox{otherwise} \end{array}\right. 
\end{split}.
\end{align}  
Notice, that we have omitted the dependency of $N_1$ and $N_2$ on $\epsilon$, $m$ and $\theta$, since, unlike $k_j$, these remain constant throughout the construction of ${\bf{A_{m}}}\left(\theta\right)$.  

The resulting Xampling scheme is depicted in Fig. \ref{Fig:05}.   Based on~\cite{Tur01}, we propose a simple mechanism for obtaining the Fourier coefficients in each individual element:  a linear transformation, ${\bf{W_m}}$, is applied to point-wise samples of the signal, taken at a sub-Nyquist rate, after filtering it with an appropriate kernel, $s_m^*\left(-t\right)$, such as the Sum of Sincs.  In this scheme, while we do need to extract larger number of samples at the output of each element, as $K_m>K$, we avoid the use of complicated analog kernels as in Section \ref{sec:52}.  Furthermore, as we show in Section \ref{sec:09}, in an actual imaging scenario good approximation is obtained with just a small sampling overhead.  
\begin{figure}
\begin{minipage}[b]{1.00\linewidth}
\centering
 \centerline{\includegraphics[width=9.2cm]{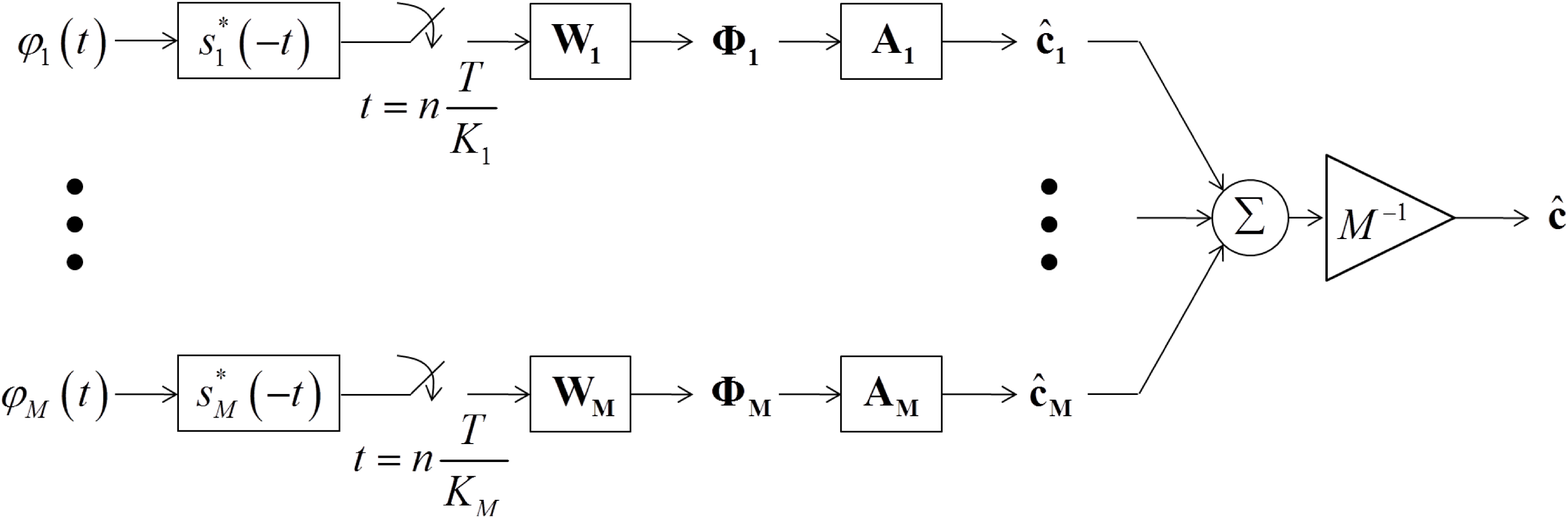}}
\end{minipage}
\caption{Xampling scheme utilizing Fourier samples of $\varphi_m\left(t\right)$.}
\label{Fig:05}
\end{figure}
\section{Signal Reconstruction} 
\label{sec:07}
So far we derived our approach for extracting the parameters $\left\{t_l,b_l\right\}_{l=1}^L$ which determine $\Phi\left(t;\theta\right)$ from sub-Nyquist samples, taken at the individual receiver elements.  In this section we focus on the  reconstruction of $\Phi\left(t;\theta\right)$ from these parameters.  Once  $\Phi\left(t;\theta\right)$ is constructed for multiple values of $\theta$, a two-dimensional image may be formed, by applying standard post-processing techniques: first, $\Phi\left(t;\theta\right)$'s envelope is extracted using the Hilbert transform~\cite{Shiavi01};  logarithmic compression is then applied to each envelope, resulting in a corresponding scanline;  finally, all scanlines are interpolated onto a two-dimensional grid.  Having obtained the parametric representation of $\Phi\left(t;\theta\right)$, the first two steps may be calculated only within the support of the recovered signal.  

In Section \ref{sec:71} we describe the reconstruction of $\Phi\left(t;\theta\right)$ from its estimated parameters, while generalizing the model proposed in \eqref{E:09}: we assume that the detected signals are additionally parametrized by unknown carrier phases of the reflected pulses,  and show that the Xampling approach allows estimation of these unknown phases.  

In Section \ref{sec:72} we propose an alternative approach for reconstructing $\Phi\left(t;\theta\right)$, using CS methodology.   

\subsection{Signal Reconstruction Assuming Unknown Carrier Phase}
\label{sec:71}
Consider the signal defined in \eqref{E:01}.  Modeling a signal of physical nature, it is obviously real-valued, implying that $a_{l,m}$ are real.  Consequently, by \eqref{E:21}, $b_l$ must also be real-valued. However, when we apply spectral analysis techniques aimed at solving the system formulated in \eqref{E:47}, there is generally no constraint that ${\bf{b}}$ be real-valued.  Indeed, solving it for samples obtained using our proposed Xampling schemes, the resulting coefficients are complex, with what appears to be random phases.  In fact, a similar phenomenon is observed when solving \eqref{E:08} for samples taken from the individual signals, $\varphi_m\left(t\right)$, as proposed in \cite{Tur01}.  Below we offer a physical interpretation of the random phases, by generalizing the model proposed in \eqref{E:09}.  The result is a closed-form solution for reconstructing the estimated signal, using the complex coefficients.  When applied, a significant improvement is observed, comparing the envelope of the reconstructed signal, with that of the original signal.  

The ultrasonic pulse $h\left(t\right)$ may be modeled by a baseband waveform, $g\left(t\right)$, modulated by a carrier at frequency $f_0$: $h\left(t\right)=g\left(t\right)\cos\left(\omega_0 t+\beta\right)$, where $\omega_0=2\pi f_0$ and $\beta$ is the phase of the carrier.  The model proposed in \eqref{E:09}, just like the one in \eqref{E:01}, assumes the detected pulses to be exact replicas of $h\left(t\right)$.  However, a more accurate assumption is that each reflected pulse undergoes a phase shift, based upon the relative complex impedances involved in its reflection~\cite{Brown01}.  We thus propose to approximate the beamformed signal as: 
\begin{align}\label{E:48}    
\begin{split}
\Phi\left(t;\theta \right)=\sum_{l=1}^L{\left|b_l\right| g\left(t-t_l\right)\cos\left(\omega_0\left(t-t_l\right)+\beta_l\right)},
\end{split}
\end{align}     
$\beta_l$ being an unknown phase.  The $j$th Fourier series coefficient of $\Phi\left(t;\theta\right)$ is now given by  
\begin{align}
\label{E:49}
\begin{split}
c_j&=\frac{1}{T}\int_{0}^{T}{\sum_{l=1}^L{\left|b_l\right| g\left(t-t_l\right)\cos\left(\omega_0\left(t-t_l\right)+\beta_l\right)}e^{-i\frac{2\pi}{T}k_jt}dt}\\
=& \frac{1}{2T}\sum_{l=1}^L{|b_l|\left(e^{i\beta_l}G\left(\omega_j-\omega_0\right)+e^{-i\beta_l}G\left(\omega_j+\omega_0\right)\right)e^{-i\omega_jt_l}},
\end{split}
\end{align} 
where $G\left(\omega\right)$ is the CTFT of $g\left(t\right)$ and $\omega_j=\frac{2\pi}{T}k_j$.  

Let $g\left(t\right)$  be approximated as a Gaussian with variance $\sigma^2$ and assume that  $k_j\geq 0$.  It is readily seen that 
\begin{align}
\label{E:50}
\begin{split}
\left|\frac{G\left(\omega_j+\omega_0\right)}{G\left(\omega_j-\omega_0\right)}\right|=e^{-2\sigma^2\omega_j\omega_0} 
\end{split}.
\end{align}
We can then choose
\begin{align}
\label{E:51}
\begin{split}
k_j\geq \frac{5T}{4\pi\sigma^2\omega_0} 
\end{split},
\end{align}
so that
\begin{align}
\label{E:52}
\begin{split}
\left|\frac{G\left(\omega_j+\omega_0\right)}{G\left(\omega_j-\omega_0\right)}\right|<10^{-2}
\end{split}.
\end{align}
This allows \eqref{E:49} to be approximated as
\begin{align}
\label{E:53}
\begin{split}
c_j\approx& \frac{1}{2T}G\left(\omega_j-\omega_0\right)\sum_{l=1}^L{|b_l|e^{i\beta_l}e^{-i\frac{2\pi}{T}k_jt_l}},
\end{split}
\end{align}
and additionally 
\begin{align}
\label{E:54}
\begin{split}
H\left(\omega_j\right)\approx \frac{1}{2} e^{i\beta}G\left(\omega_j-\omega_0\right).
\end{split}
\end{align} 
Combining \eqref{E:53} and \eqref{E:54}, we get
\begin{align}
\label{E:55}
\begin{split}
c_j&\approx \frac{1}{T}H\left(\frac{2\pi}{T}k_j\right)\sum_{l=1}^L{b_l e^{-i\frac{2\pi}{T}k_jt_l}},
\end{split}
\end{align}
where we define $b_l=|b_l|e^{i\left(\beta_l-\beta\right)}$. 

Denoting by ${\bf{c}}$ the length $K$ vector, with $c_j$ as its $j$th element, the last result may be brought into the exact same matrix form written in \eqref{E:47}.  However, now we expect the solution to extract complex coefficients, of which phases correspond to the unknown phase shifts of the reflected pulses, $\angle b_l = \beta_l-\beta$.  Having obtained the complex coefficients, we may now reconstruct $\Phi\left(t;\theta\right)$ according to \eqref{E:48}, and then proceed with standard post-processing techniques.  The constraint imposed in \eqref{E:51} is mild, considering nominal ultrasound parameters.  Assuming, for instance, $T=210\mu\mbox{sec}$, $f_0=3\mbox{MHz}$, and $\sigma=630\mbox{nsec}$, we must choose $k_j\geq 12$.  The requirement that ${\bf{H}}$ of \eqref{E:28} be invertible, already imposes a stronger constraint on $k_j$, the $j$th Fourier coefficient, since $H\left(\frac{2\pi}{T}k_j\right)$ drops below $-3\mbox{dB}$ for $\left|k_j-630\right|>44$.          
\subsection{CS Approach for Signal Reconstruction} 
\label{sec:72}

Throughout the previous sections, we addressed the problem of ultrasound signal reconstruction, within the FRI framework.  As shown in \cite{Vetterli01}, for various FRI problems, the relationship between the unknown signal parameters and its subset of Fourier series coefficients takes the form of a spectral analysis problem.  The latter is then typically solved by applying techniques such as annihilating filter~\cite{Stoica01} or matrix pencil~\cite{Tapan01}.   In this section, we consider an alternative approach for reconstructing the signal defined in \eqref{E:09}, based on CS methodology~\cite{Candes01,Davenport01}.  

Assume that the time delays $\left\{t_l\right\}_{l=1}^L$ in \eqref{E:48} are quantized with a $\Delta_s$ quantization step, such that $t_l=q_l\Delta_s$, $q_l\in\mathbb{Z}$.  Using \eqref{E:55}, we may write the Fourier series coefficients of $\Phi\left(t;\theta\right)$ as:
\begin{align}
\label{E:34}
c_j&\approx\frac{1}{T}H\left(\frac{2\pi}{T}k_j\right)\sum_{l=1}^L{b_l e^{-i\frac{2\pi}{T}\Delta_s k_j q_l}}.
\end{align}
Let $N$ be the ratio $\lfloor T/\Delta_s\rfloor$.  Then \eqref{E:34} may be expressed in the following matrix form: 
\begin{equation}
\label{E:35}
\begin{split}
{\bf{c}}&\approx\frac{1}{T}{\bf{H\hat{V}}x}={\bf{Ax}},\\
\end{split}
\end{equation}
where ${\bf{H}}$ is the $K\times K$ diagonal matrix with $H\left(\frac{2\pi}{T}k_j\right)$ as its $j$th diagonal element, and ${\bf{x}}$ is a length $N$ vector, whose $j$th element equals $b_l$ for $j=q_l$, and $0$ otherwise.  Finally, ${\bf{\hat{V}}}$ is a $K\times N$ matrix, formed by taking the set $\kappa$ of rows from an $N\times N$ FFT matrix.

The formulation obtained in \eqref{E:35}, is a classic CS problem, where our goal is to reconstruct the $N$-dimensional vector ${\bf{x}}$, known to be $L$-sparse, with $L\ll N$, based on its projection onto a subset of $K$ orthogonal vectors, represented by the rows of ${\bf{A}}$.  This problem may be solved by various CS methods, as long as the sensing matrix ${\bf{A}}$ satisfies desired properties such as the Restricted Isometry Property (RIP) or coherence.  

In our case, ${\bf{A}}$ is formed by choosing $K$ rows from the Fourier basis.  Selecting these rows uniformly at random it may be shown that if 
\begin{equation}
\label{E:38}
K\geq CL\left(\log N\right)^4,
\end{equation}
for some positive constant $C$, then ${\bf{A}}$ obeys the RIP with large probability \cite{Rudelson01}.  As readily seen from \eqref{E:38}, the resolution of the grid, used for evaluating $\left\{t_l\right\}_{l=1}^L$, directly effects the RIP.  Recall that, by applying spectral analysis methods, one may reconstruct $\bf{x}$ from a minimal number of $2L$ samples, if it is indeed $L$-sparse.  However, these samples must be carefully chosen.  Using matrix pencil, for instance, the sensing vectors must be  consecutive.  Moreover, in any practical application, the measured data will be corrupted by noise, forcing us to use oversampling.  In contrast, the bound proposed in \eqref{E:38} regards random selection of the sensing vectors.  Additionally, applying the CS framework, we may effectively cope with the more general case, of reconstructing ${\bf{x}}$ which is not necessarily $L$-sparse.  

\section{Comparison between Recovery Methods}
\label{sec:08}
In this section, we provide results obtained by applying three recovery algorithms to ultrasound signals which were simulated using the {\it{Field II}} program~\cite{Jensen03}.  The evaluation was performed based on multiple beamformed signals, each calculated along the $\hat{z}$ axis ($\theta=0$) for a random phantom realization.  The phantom comprised $L$ strong reflectors, distributed along the $\hat{z}$ axis, and multiple additional reflectors, distributed throughout the entire imaging medium.  A measurement vector was obtained by projecting the beamformed signal onto a subset of its Fourier series coefficients.  Finally, each algorithm was evaluated for its success in recovering the strong reflectors' positions from the vector of measurements.  The first two algorithms which were evaluated were matrix pencil~\cite{Tapan01} and total least-squares approximation, enhanced by Cadzow's iterated algorithm \cite{Blu01}.  Both algorithms may be considered spectral analysis techniques.  The third algorithm was Orthogonal Matching Pursuit (OMP) \cite{Tropp01}, which is a CS method.  

The simulation setup is depicted in Fig. \ref{Fig:06}.  We created an aperture comprising 64 transducer elements, with central frequency $f_0=3.5\mbox{MHz}$.  The width of each element, measured along the ${\hat{x}}$ axis, was $c/{f_0}=0.44\mbox{mm}$, and the height, measured along the ${\hat{y}}$ axis, was $5\mbox{mm}$.  The elements were arranged along the $\hat{x}$ axis, with a $0.05\mbox{mm}$ kerf.  The transmitted pulse was simulated by exciting each element with two periods of a sinusoid at frequency $f_0$, where the delays were adjusted such that the transmission focal point was at depth $r=70\mbox{mm}$.  Additionally, Hanning apodization was used during transmission, by applying an appropriate excitation power to each element.    

\begin{figure}
\begin{minipage}{1\linewidth}
\centering
\centerline{\includegraphics[height=3.5cm]{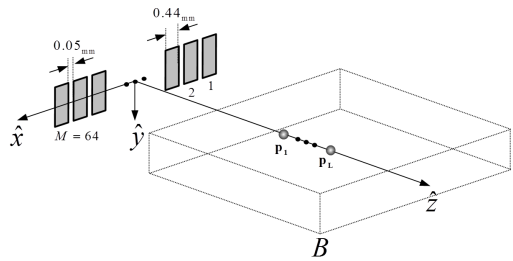}}
\end{minipage}
\caption{{\it{Field II}} simulation setup: $M=64$ elements are aligned along the $\hat{x}$ axis with a $0.05\mbox{mm}$ kerf.  The width of each element is $0.44\mbox{mm}$.  Speckle pattern is simulated by randomly distributing $10^5$ point reflectors within the box $B$.  Additionally, $L=6$ point reflectors are aligned along the $\hat{z}$ axis, also within the boundaries of the box.  The pulse is transmitted along the $\hat{z}$ axis, and the beamformed signal is constructed along the same line.}
\label{Fig:06}
\end{figure}

In each iteration, we constructed a random phantom, for which we simulated the beamformed signal.  The phantom was constructed in two stages.  We first created a speckle phantom, by drawing positions of $10^5$ point reflectors uniformly, at random, within the three-dimensional box $B=\left\{\left(x,y,z\right): |x| \leq 25\mbox{mm}, |y|  \leq 5\mbox{mm}, |z-60| \leq 30\mbox{mm}\right\}$. The corresponding amplitudes were also drawn randomly, with zero-mean and unit-variance Normal distribution.  We then generated a signal phantom, by drawing positions of $L=6$ point reflectors, $\left\{{\bf{p_l}}\right\}_{l=1}^L$, with $x_l=y_l=0$ and $z_l$ uniformly distributed in the interval $\left[35\mbox{mm},85\mbox{mm}\right)$.  These reflectors were assigned identical amplitudes, which were adjusted according to the SNR requirement, in the following manner:  for each of the two phantoms, we simulated the beamformed signal, acquired along $\theta=0$ following pulse transmission in the same direction.   Denoting the beamformed signal obtained for the first (speckle) phantom by $n\left(t;\theta=0\right)$ and that obtained for the second (signal) phantom by $\Phi\left(t;\theta=0\right)$, we defined
\begin{align}
\label{E:SNR1}
{\mbox{SNR}}=10\log_{10}\frac{\int_0^T{\left|\Phi\left(t;\theta=0\right)\right|^2dt}}{\int_0^T{\left|n\left(t;\theta=0\right)\right|^2dt}}.
\end{align} 
The amplitudes of the reflectors comprising the second phantom were modified, such that \eqref{E:SNR1} complied with the desired SNR value.  After this calibration, we combined the two phantoms into a single one, for which we generated an individual beamformed signal realization.  The detected signals and the resulting beamformed signal were simulated at sampling rate $f_s=100\mbox{MHz}$.  Since the spectrum of the detected pulses decayed to $-50\mbox{dB}$ at $\approx 6\mbox{MHz}$, this rate was far beyond Nyquist.  Hanning apodization was used for constructing the beamformed signal, by applying appropriate weights to the detected signals.  This type of apodization may be easily implemented with both our Xampling schemes, by replacing the average in \eqref{E:24} by a weighted one. 

Fig. \ref{Fig:07} illustrates the method by which we simulated a realization of the noisy beamformed signal.  This image was obtained by applying standard imaging techniques to an individual phantom.  We are interested in recovering the strong reflections aligned along the $\hat{z}$ axis.  The corresponding beamformed signal was corrupted by speckles, originating in the multiple point reflectors scattered throughout the medium.  The phantom was calibrated such that the SNR of the beamformed signal along $\theta=0$, defined in \eqref{E:SNR1}, was $15\mbox{dB}$.            

\begin{figure}
\begin{minipage}{1\linewidth}
\centering
\centerline{\includegraphics[height=5.0cm]{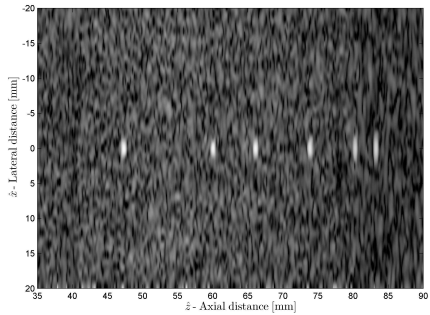}}
\end{minipage}
\caption{Image obtained by applying standard imaging techniques to an individual phantom realization.  Our goal is to recover the $L=6$ strong reflectors aligned along the $\hat{z}$ axis.  $10^5$ point reflectors were distributed in the imaging plain, resulting in echoes which corrupted the detected signals.  In the ultrasound image, these reflections are manifested in a speckle pattern.   The phantom was calibrated such that the SNR of the beamformed signal, calculated along $\hat{\theta=0}$, defined in \eqref{E:SNR1}, was $15\mbox{dB}$.}
\label{Fig:07}
\end{figure}

Having generated the beamformed signal, we obtained a measurement vector,  by projecting the signal onto a subset of its $K$ Fourier series coefficients, where $K=2\lceil \eta L \rceil+1$, and $\eta>1$ is the desired oversampling factor.  For the spectral analysis techniques, we chose the coefficients consecutively, around $k_0=\lceil f_0 T \rceil$.  OMP was tested using both this selection of coefficients, and a random selection, taken such that $H\left(\frac{2\pi}{T}k_j\right)$ is above $-2\mbox{dB}$.  With this selection, we obtain samples which are better spread in the frequency domain.  We emphasize, that the coefficients were drawn once, for each choice of $\eta$. An additional degree of freedom, using the OMP method, regards the density of the reconstruction grid, determined by $N$.  We set  $N=1860$, complying with a sampling frequency $f_s=20\mbox{MHz}$, of order typically used in imaging devices.  

Recovery was evaluated based on the estimated time delays.  These were compared to the delays associated with the known reflector positions, $t_l=2z_l/c$. At the $i$th iteration, we examined, for each algorithm, all possible matches between actual delays $\left\{t_l\right\}_{l=1}^L$, and  estimated delays $\left\{\hat{t}_l\right\}_{l=1}^L$.  Of all possible permutations (a total number of $L!$), we selected the one for which the  number of matches, achieving error smaller than the width of $h\left(t\right)$, was maximal.  Denoting this number by $S_i^{(q)}$, $q\in\left\{1,2,3,4\right\}$ corresponding to the evaluated method, we estimate the probability of recovery by the $q$th method as  
\begin{equation}
\label{E:39}
P^{(q)}=\frac{1}{LI}\sum_{i=1}^I{S_i^{(q)}},
\end{equation}
where $I$ is the total number of iterations, set to 500 in our simulation.  We note that all reconstruction algorithms require that we first calculate $H\left(\frac{2\pi}{T}k_j\right)$.  For this purpose, we simulated the signal beamformed along $\theta=0$, for a phantom which comprised a single reflector at the transmission focal point $\left(x,y,z\right)=\left(0,0,70\mbox{mm}\right)$.  We used the detected signal, depicted in Fig. \ref{Fig:08}, for calculating $H\left(\frac{2\pi}{T}k_j\right)$.      

\begin{figure}
\begin{minipage}{1\linewidth}
\centering
\centerline{\includegraphics[height=4.0cm]{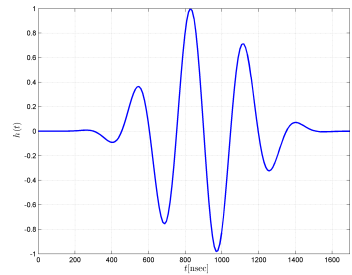}}
\end{minipage}
\caption{$h\left(t\right)$ evaluated from the beamformed signal, calculated for a single reflector using {\it{Field II}} simulator.  The reflector was positioned at the transmission focal point.}
\label{Fig:08}
\end{figure}

\begin{figure*}
\begin{minipage}{.23\linewidth}
\centering
\centerline{\includegraphics[height=3.7cm]{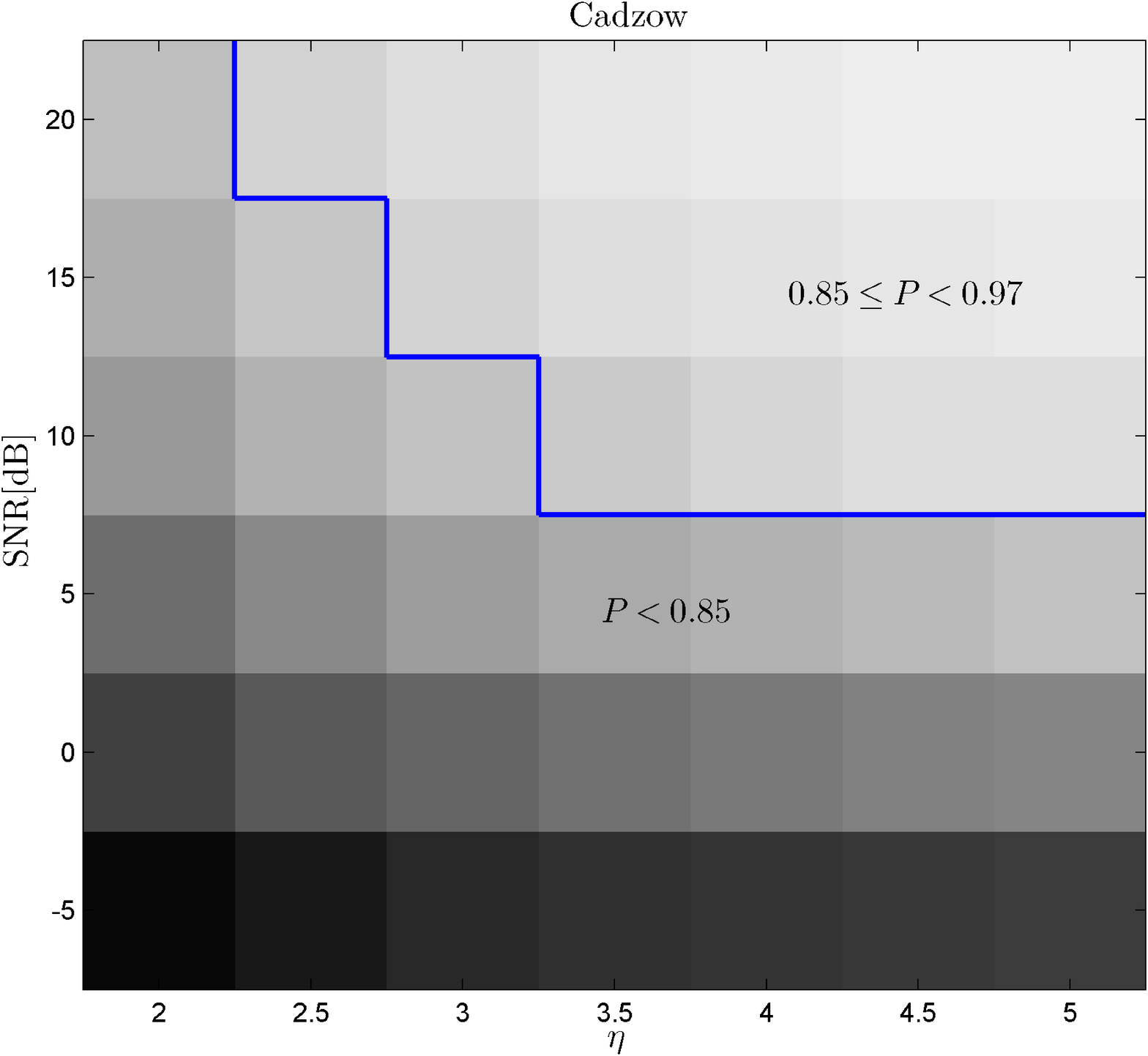}}
\centerline{(a)}
\end{minipage}
\begin{minipage}{.23\linewidth}
\centering
\centerline{\includegraphics[height=3.7cm]{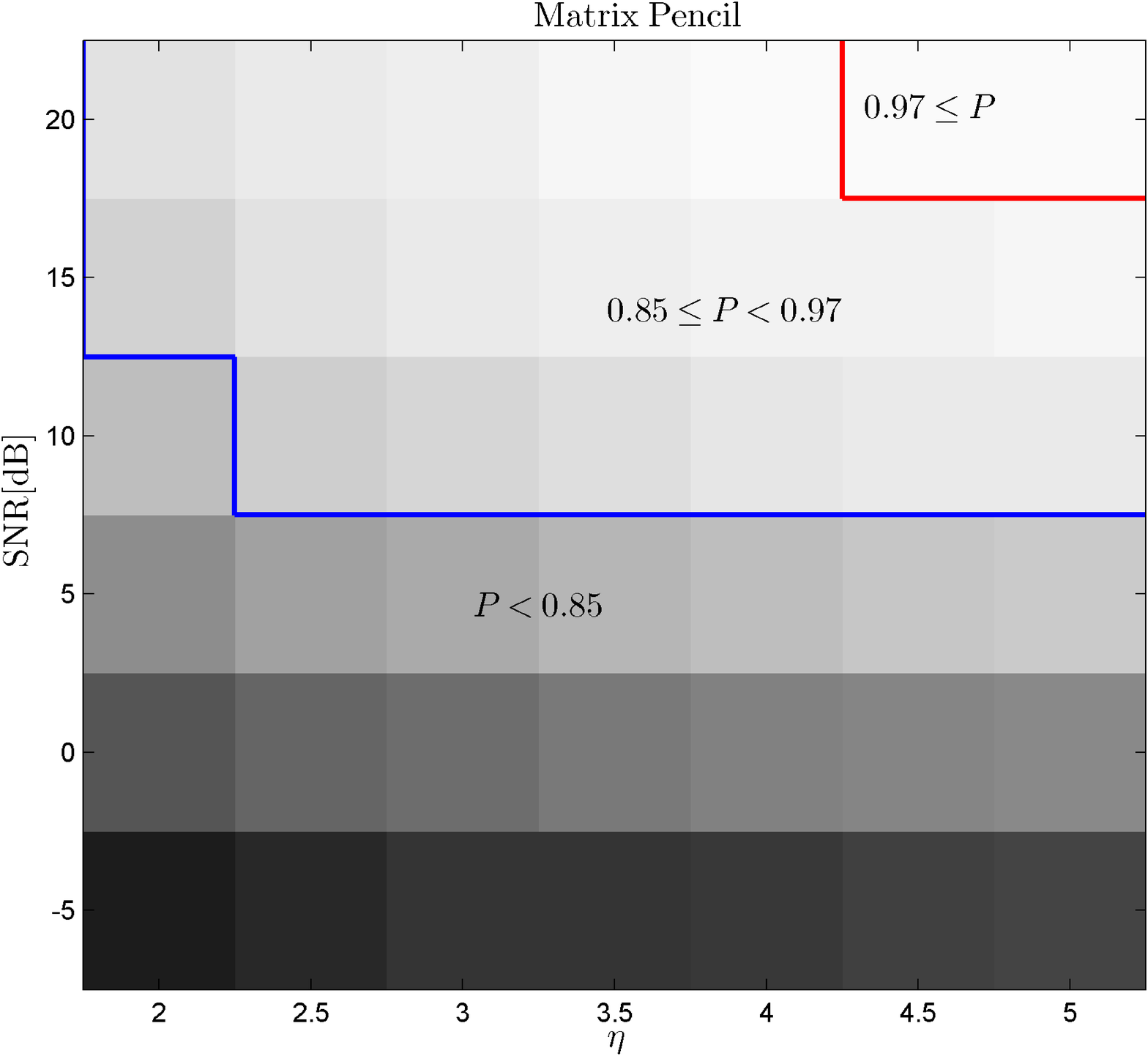}}
\centerline{(b)}
\end{minipage}
\begin{minipage}{.23\linewidth}
\centering
\centerline{\includegraphics[height=3.7cm]{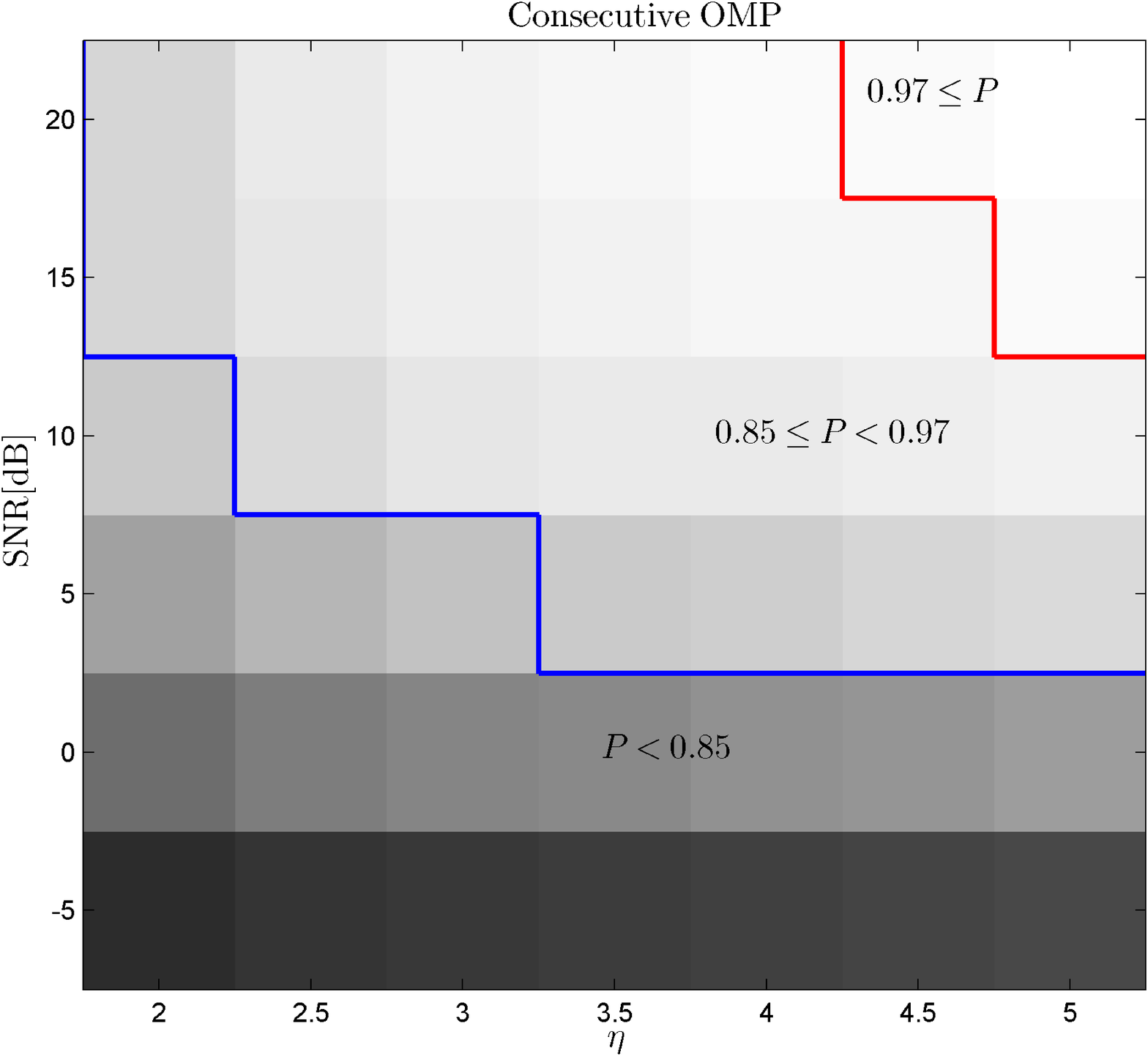}}
\centerline{(c)}
\end{minipage}
\begin{minipage}{.26\linewidth}
\centering
\centerline{\includegraphics[height=3.7cm]{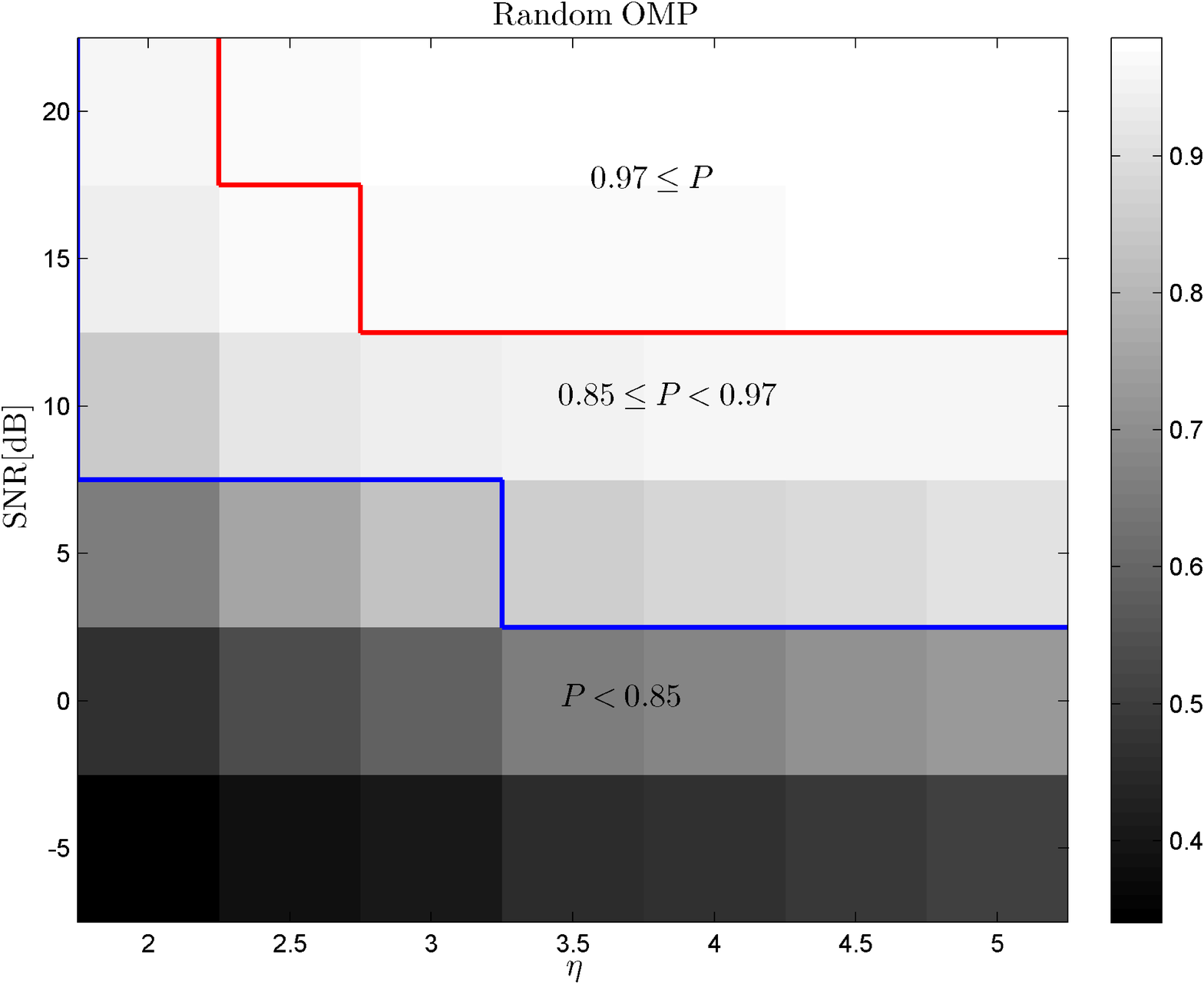}}
\centerline{(d)}
\end{minipage}
\caption{Probability of reconstruction vs. SNR and oversampling factor, $\eta$, using four methods: (a) Total least-squares, enhanced by Cadzow's iterated algorithm, (b) matrix pencil, (c) OMP with consecutive Fourier series coefficients,  (d) OMP with Fourier series coefficients randomly distributed, such that $H\left(\frac{2\pi}{T}k_j\right)$ is above $-2\mbox{dB}$, $\forall k_j\in \kappa$.  Signals were simulated using {\it{Field II}} program, where SNR is defined in \eqref{E:SNR1}.}
\label{Fig:09}
\end{figure*}
%

The simulation results obtained for multiple combinations of SNR and oversampling factor are illustrated in Fig. \ref{Fig:09}.  The calculated recovery probabilities are represented by gray-levels, where a common color-bar was used for all plots.  For clarity, we plotted a line separating between probabilities lower than 0.85 and probabilities above 0.85, and a line separating between probabilities lower than 0.97 and probabilities above 0.97. Of the two spectral analysis techniques, matrix pencil appears preferable, as it obtains high probability values over a wider range of SNR and oversampling.  Both OMP methods outperformed the spectral analysis ones, with an obvious advantage to random OMP.      

An additional aspect which should be taken into consideration, when choosing the reconstruction method, regards the complexity of the Xampling hardware.  Using the Xampling scheme proposed in \cite{Tur01}, random selection  of Fourier series coefficients will increase the hardware complexity: in such case, the sampling kernel, e.g. SoS, must be specifically designed for the choice of coefficients.  This is in contrast with the relatively simple kernel, applied for a consecutive choice of coefficients.  On the other hand, the Xampling scheme proposed in \cite{Gedalyahu01} is practically invariant to the manner in which the coefficients are selected.  

\section{Experiments on Cardiac Ultrasound Data}	
\label{sec:09}
\begin{figure*}
\begin{minipage}{.33\linewidth}
\centering
\centerline{\includegraphics[width=5.0cm]{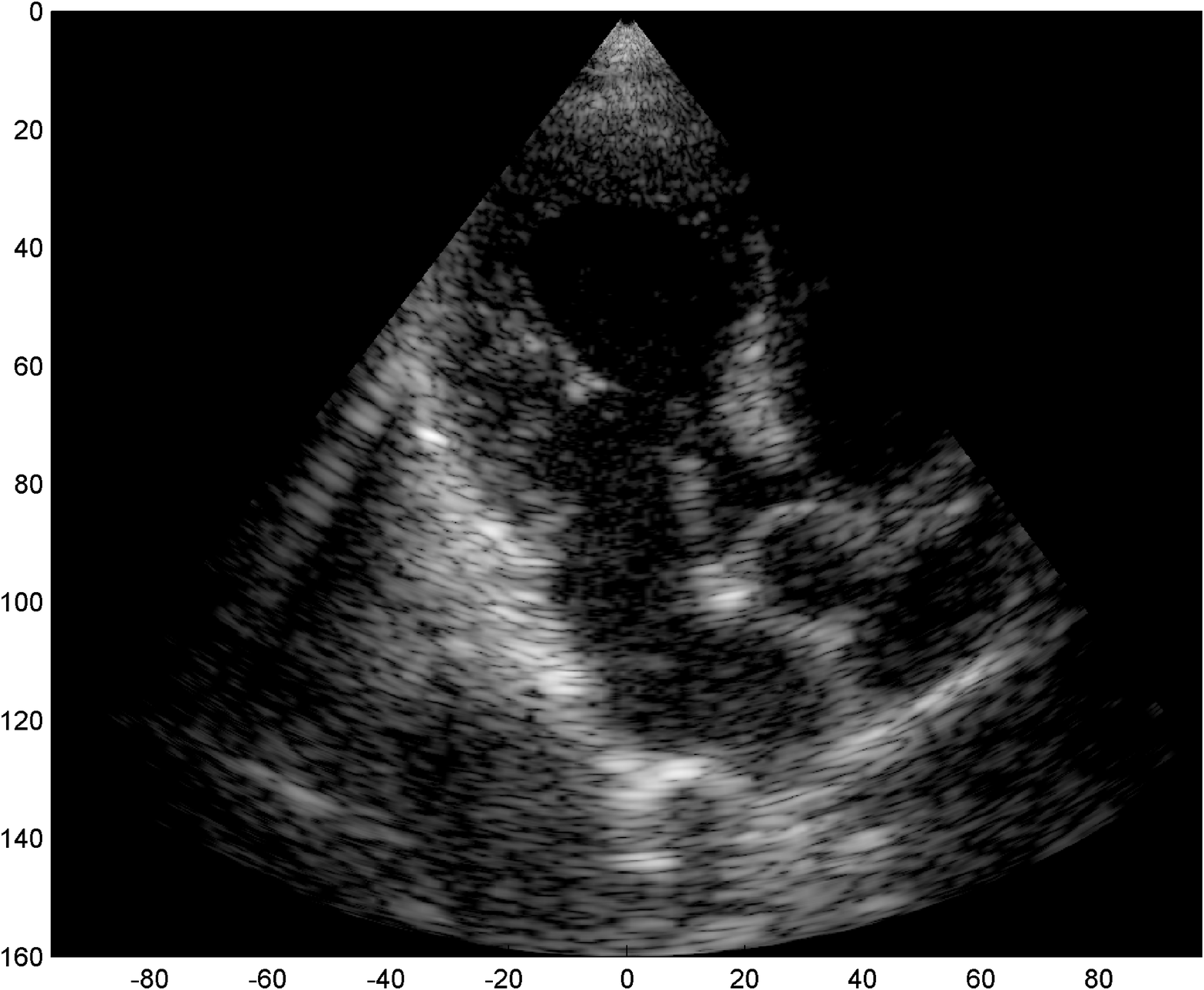}}
\centerline{(a)}
\end{minipage}
\begin{minipage}{.33\linewidth}
\centering
\centerline{\includegraphics[width=5.0cm]{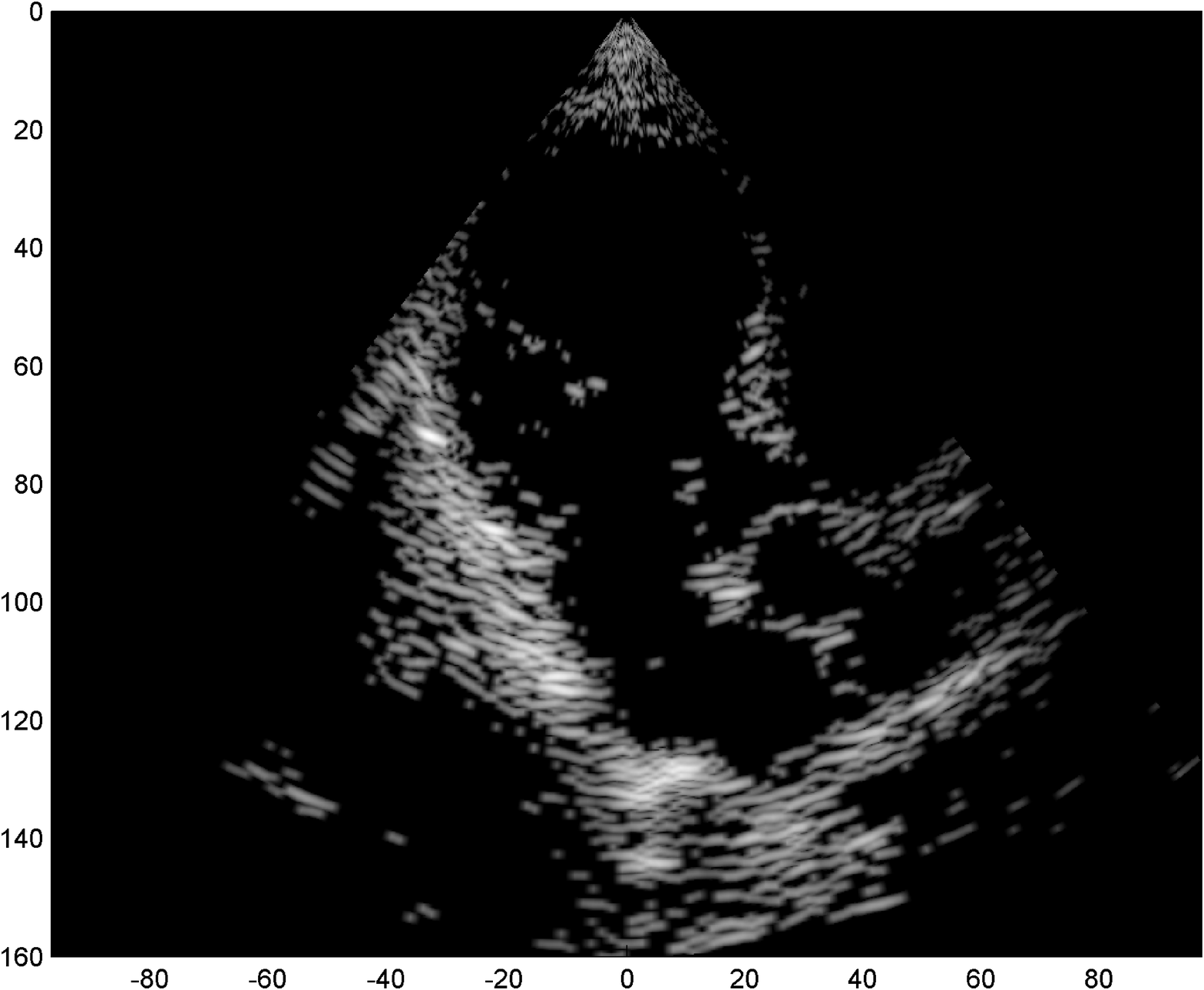}}
\centerline{(b)}
\end{minipage}
\begin{minipage}{.34\linewidth}
\centering
\centerline{\includegraphics[width=5.0cm]{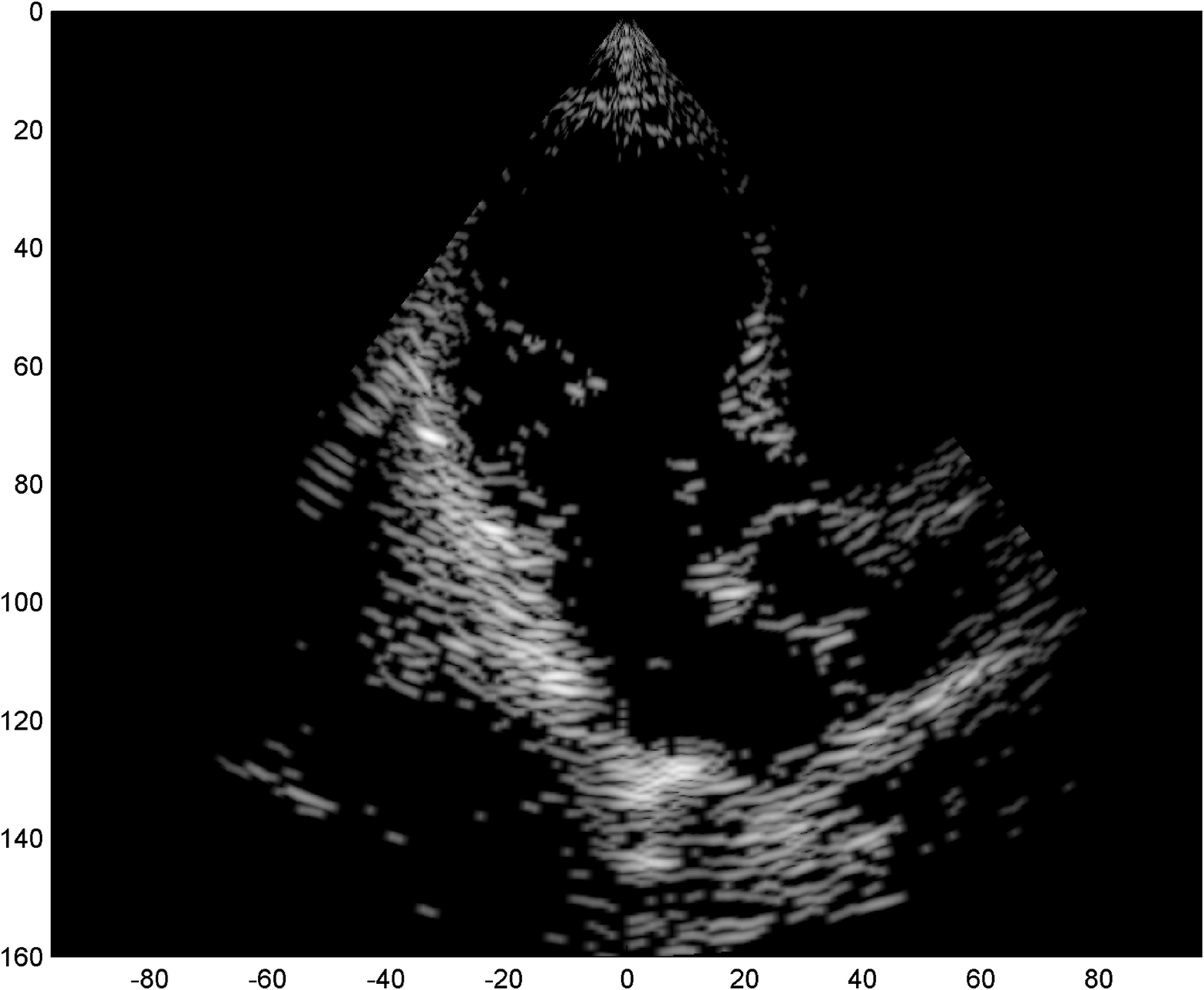}}
\centerline{(c)}
\end{minipage}
\caption{Cardiac images generated by Xampling and using traditional methods. (a) standard beamforming applied to data sampled at the Nyquist-rate.  (b) applying the non-approximated Xampling scheme of Fig.~\ref{Fig:04}.  (c) applying the final Xampling scheme of Fig.~\ref{Fig:05}.}
\label{Fig:10}
\end{figure*}


In this section, we examine results obtained by applying our Xampling schemes, illustrated in Figs.~\ref{Fig:04}~and~\ref{Fig:05}, to raw RF data, acquired and stored for cardiac images of a healthy consenting volunteer.  The acquisition was performed using a GE breadboard ultrasonic scanner of 64 acquisition channels.  The transducer employed was a 64-element phased array probe, with $2.5\mbox{MHz}$ central frequency,  operating in second harmonic imaging mode: 3 half cycle pulses are transmitted at $1.7\mbox{MHz}$, resulting in a signal characterized by a rather narrow bandpass bandwidth, centered at $1.7\mbox{MHz}$.  The corresponding second harmonic signal, centered at $3.4\mbox{MHz}$, is then acquired.  The signal detected in each acquisition channel is amplified and digitized at a sampling-rate of $50\mbox{MHz}$.   Data from all channels were acquired along 120 beams, forming a $60^{\circ}$ sector, where imaging to a depth of $z=16\mbox{cm}$, we have $T=207{\mu\mbox{sec}}$.  The imaging results are illustrated in Fig. \ref{Fig:10}. 

The first image (a) was generated using the standard technique, applying beamforming to data first sampled at the Nyquist-rate, and then down-sampled, exploiting its limited essential bandwidth.  For a single scanline, sampling at $50 {\mbox{MHz}}$, we acquire $10389$ real-valued samples from each element, which are then down-sampled, to $1662$ real-valued samples, used for beamforming.  The resulting image is used as reference, where our goal is to reproduce the macroscopic reflectors observed in this image with our Xampling schemes.   

We begin by applying the scheme illustrated in Fig. \ref{Fig:04}, utilizing the analog kernels defined in \eqref{E:27}.   Modulation with the kernels is simulated digitally.   Assuming $L=25$ reflectors, and using two-fold oversampling, $\kappa$ comprises $K=100$ consecutive indices.  With such selection, the corresponding frequency samples practically cover the essential spectrum of $h\left(t\right)$.  Since each sample is complex, we get an eight-fold reduction in sample-rate.  Having estimated the Fourier series coefficients of $\Phi\left(t;\theta\right)$, we obtain its parametric representation by solving \eqref{E:35} using OMP.  We then reconstruct $\Phi\left(t;\theta\right)$ according to \eqref{E:48}, that is we apply phase shifts to the modulated pulses,  based on the extracted coefficients' phases.  The resulting image (b) depicts the strong perturbations observed in (a).  Moreover, isolated reflectors at the proximity of the array ($z\approx 6{\mbox{cm}}$) remain in focus.  

We next apply the approximated scheme, illustrated in Fig. \ref{Fig:05}: for every $k_j\in\kappa$, $1\leq m \leq M$ and $\theta$, we find $N_1$ and $N_2$ of \eqref{E:29} such that $\rho^2 \approx 0.95$.  This process is performed numerically, off-line, based on our imaging setup.  Consequently, we construct $\left\{{\bf{A_m}}\right\}_{m=1}^M$ off-line, according to \eqref{E:33}.  Choosing this level of approximation, we end up with a seven-fold reduction in sample rate, where, for the construction of a single scanline, an average of $116$ complex samples must be taken from each element.  We point out that in this scenario, the maximal number of samples, taken from certain elements, reaches $133$ for specific values of $\theta$.  Thus, if a common rate is to be used for all sensors, for all values of $\theta$, we may still achieve a six-fold reduction in sample rate.  As before, we use OMP in order to obtain  $\Phi\left(t;\theta\right)$'s parametric representation, and reconstruct it based on our generalized FRI model proposed in \eqref{E:48}.  The resulting image (c) appears very similar to (b). 

Table \ref{Tab:01} gathers SNR values, calculated for the beamformed signals estimated using both our Xampling schemes, after envelope detection with the Hilbert transform.  The values were calculated with respect to the envelopes of the beamformed signals, obtained by standard imaging.  Explicitly, let $\Phi\left(t;\theta_i\right)$ denote the beamformed signal obtained by standard beamforming along the direction $\theta_i$, $i=1,2,...,I$, let $\hat{\Phi}\left(t;\theta_i\right)$ denote the beamformed signal reconstructed from the parameters recovered by compressed beamforming along the same direction, and let $H\left(\cdot\right)$ denote the Hilbert transform.  For the set of $I=120$ scanlines, we defined the SNR as
\begin{align}
\label{E:SNR2}
\mbox{SNR}=10\log_{10}\frac{\sum_{i=1}^I{\int_0^T{\left|\mbox{H}(\Phi\left(t;\theta_i\right))\right|^2dt}}}{\sum_{i=1}^I{\int_0^T{\left|\mbox{H}(\hat{\Phi}\left(t;\theta_i\right))-\mbox{H}(\Phi\left(t;\theta_i\right))\right|^2dt}}}.
\end{align}
This calculation was repeated when reconstructing the signals without the random phase assumption, proposed in Section~\ref{sec:71}.  For the latter case, reconstruction of a real-valued $\hat{\Phi}\left(t;\theta_i\right)$, given complex coefficients, may be heuristically achieved by either ignoring the coefficients' imaginary part, or by taking their modulus.  It may be seen that, weighting over all $120$ beamformed signals, the random phase assumption achieves a relatively minor improvement ($0.1$-$0.15\mbox{dB}$) compared to reconstruction using the modulus of the coefficients.    However, when examining individual signals, we observed that, for certain values of $\theta_i$, the improvement exceeded  $1.5\mbox{dB}$.  

\begin{table}[h]
\centering
\caption{SNR in [dB] of $\Phi\left(t;\theta\right)$ obtained with the proposed Xampling schemes and three reconstruction methods}
\begin{tabular}{|p{1.2in}||p{0.85in}|p{0.85in}|} \hline
 & \multicolumn{2}{c|}{Xampling Method} \\ \cline{2-3}
Reconstruction Method & Distorted Kernels  (Fig. \ref{Fig:04}) & Approximated Scheme (Fig. \ref{Fig:05}) \\ \hline\hline
Random Phase 	& {6.47} &  {5.89} \\ \hline
Real part of Residues 	& {4.59} &  {4.03} \\ \hline
Modulus of Residues 	& {6.32} &  {5.79} \\ \hline
\end{tabular}
\label{Tab:01}.
\end{table}

We emphasize, that the calculated SNR values provide a useful measure for quantitatively comparing the different Xampling and reconstruction approaches.  However, they are of smaller value when attempting to evaluate the overall performance of Xampling, compared to standard imaging: recall that our scheme is aimed at reproducing only strong pulses, reflected from macroscopic reflectors.  The reference signal, on the other hand, generated by standard technique, already contains the  additional speckle component, caused by multiple microscopic perturbations.  A possible approach for evaluating the overall performance of either Xampling scheme, would be to examine its success rate in recovering strong reflections, detected by standard beamforming.  For this purpose, we tracked the $L$ strongest local maxima in each beamformed signal.  If the Xampling scheme recovered a pulse within the range of $1.2\mbox{mm}$ from a certain maximum, we say that this maximum was successfully detected.  Certain pulses, detected by Xampling, may match more than one maximum in the beamformed signal.  In such case, we choose the one-to-one mapping which achieves smallest MSE.  Applying this evaluation method to signals Xampled using our approximated scheme, and reconstructed with the random phase assumption, we conclude that the reconstruction successfully retrieves $70.4 \%$ of the significant maxima, with standard deviation of the error being approximately $0.42\mbox{mm}$. 
 
\section{Conclusion}
\label{sec:10}
In this work, we generalized the Xampling method proposed in \cite{Tur01}, to a scheme applied to an array of multiple receiving elements, allowing reconstruction of a two-dimensional ultrasound image. At the heart of this generalization was the proposal that the one-dimensional Xampling method derived in \cite{Tur01} be applied to signals obtained by beamforming.  Such signals exhibit enhanced SNR, compared to the individual signals detected by the array elements.  Moreover, they depict reflections which originate in a much narrower sector, than that initially radiated by the transmitted pulse.  A second key observation, which made our approach feasible, regarded the integration of the beamforming process into the filtering part of the Xampling scheme.  

The first approach we purposed comprised multiple modulation and integration channels, utilizing analog kernels.  We next showed that the parametric representation of the beamformed signal may be well approximated, from projections of the detected signals onto appropriate subsets of their Fourier series coefficients.  The contribution of our schemes regards both the reduction in sample rate, but additionally, the resulting reduction in the rate of data transmission from the system front-end to the processing unit.  In particular, our second approach is significant even when preliminary sampling is performed at the Nyquist-rate.  In such a case, it allows a reduction in data transmission rate, by a relatively simple linear transformation, applied to the sampled data.          


An additional contribution of our work regards the method by which we reconstruct the ultrasound signal, assumed to obey a specific FRI structure, from a subset of its frequency samples.  Rather than using traditional spectral analysis techniques, we formulate the relationship between the signal's samples to its unknown parameters as a CS problem.  The latter may be efficiently solved using a greedy algorithm such as the OMP.  We show that, in our scenario, CS is generally comparable to spectral analysis methods, managing to achieve similar success rates with sample sets of equal cardinality.  Moreover, working in a noisy regime, CS typically outperformed spectral analysis methods, provided that the frequency samples were highly spread over the essential spectrum of the signal.  Using actual cardiac data, a relatively large number of reflectors was assumed.  Consequently, by simply choosing the Fourier series coefficients consecutively, as in the spectral analysis techniques, we end up with the necessary wide distribution.  However, as shown in our simulations, CS approach inherently allows a wide distribution of samples, even when the cardinality of the sample set is small, since we are not obliged to unique configurations of samples.

A final observation discussed in our work, regards the generalization of the signal model proposed in \cite{Tur01}, allowing additional, unknown phase shifts, of the detected pulses.  We show that these shifts may be estimated by appropriate interpretation of the extracted coefficients, without changing the recovery method.   

Combining the random phase assumption with our proposed Xampling schemes and the CS recovery method, we construct two-dimensional ultrasound images, which well depict strong perturbations in the tissue, while achieving up to seven-fold reduction of sample rate, compared to standard imaging techniques.                                

\appendices
\section{beamformed signal support}

We assume $h\left(t\right)$ to be supported on $\left[0,\Delta\right)$, and that the support of $\varphi_m\left(t\right)$  is contained in $\left[0,T\right)$.  The last assumption may be justified by the fact that the pulse is transmitted at $t=0$, such that reflections may only be detected for $t\geq 0$.  Additionally, the penetration depth of the transmitted pulse allows us to set $T$, such that all reflections arriving at $t\geq T$ are below the noise level.  

For all $1\leq l\leq L$ and $1\leq m\leq M$:
\begin{align}
\label{E:40}
\begin{split}
t_{l,m}+\Delta \leq T,
\end{split}
\end{align}
Applying the relation $t_{l,m}=\tau_m\left(t_l;\theta\right)$, justified in Section~\ref{sec:51}, and using the fact that $\tau_m\left(t;\theta\right)$ is non-decreasing for $t\geq 0$ we conclude that
\begin{align}
\begin{split}
\label{E:41}
t_l&\leq  \tau_m^{-1}\left(T-\Delta;\theta\right),
\end{split}
\end{align}
$\tau_m^{-1}\left(t;\theta\right)$ being the inverse of $\tau_m\left(t;\theta\right)$.  Explicitly: 
\begin{align}
\begin{split}
\label{E:42}
\begin{array}{ll}
{\tau_m^{-1}\left(t;\theta\right)}=\frac{t^2-\gamma_m^2}{t-\gamma_m\sin\theta}, & t\geq \gamma_m.
\end{array}
\end{split} 
\end{align}
Assuming that $\Delta\ll T$, then, since \eqref{E:41} is true for every $1 \leq m \leq M$, we may write:
\begin{align}
\begin{split}
\label{E:43}
t_l \leq \min_{1\leq m\leq M}{\tau_m^{-1}\left(T;\theta\right)}.
\end{split}
\end{align}
This allows us to set the following upper bound on the support of $\Phi\left(t;\theta\right)$:
\begin{align}
\begin{split}
\label{E:44}
 T_B\left(\theta\right) = \min_{1\leq m\leq M}{\tau_m^{-1}\left(T;\theta\right)},
\end{split}
\end{align}
once again, using the assumption that $\Delta\ll T$.  From \eqref{E:42} it is readily seen that $ T_B\left(\theta\right)\leq T$, since we can always find $\gamma_m$ with sign opposite to that of $\sin\theta$, such that:
\begin{align}
\begin{split}
\label{E:45}
{\tau_m^{-1}\left(T;\theta\right)}=\frac{T^2-\gamma_m^2}{T+|\gamma_m\sin\theta|}\leq \frac{T^2-\gamma_m^2}{T}\leq T.
\end{split} 
\end{align}
Finally, by construction of $T_B\left(\theta\right)$ we see that, for all $1\leq m\leq M$, $\tau_m\left(T_B\left(\theta\right);\theta\right) \leq T$. 
 
%

\bibliographystyle{IEEEtran}

\end{document}